\documentclass[10pt]{article}

\usepackage{PRIMEarxiv}
\usepackage[symbol]{footmisc}

\usepackage{tensor}

\usepackage{setspace}
\usepackage{makeidx}
\usepackage[dvipsnames]{xcolor}
\definecolor{bluecite}{HTML}{0875b7}
\usepackage{hyperref}\hypersetup{
    colorlinks = true,
    linkcolor = bluecite,
    anchorcolor = bluecite,
    citecolor = bluecite,
    filecolor = bluecite,
    urlcolor = bluecite
   }
\usepackage[sorting=none,citestyle=numeric-comp]{biblatex}
\usepackage[T1]{fontenc}    
\usepackage{url}            
\usepackage{newpxtext}
\usepackage{booktabs}       
\usepackage{mathpazo}

\usepackage{amsfonts}       
\usepackage{nicefrac}       
\usepackage{microtype}      
\usepackage{lipsum}
\RequirePackage[]{tocloft}
\usepackage{amsthm}
\usepackage{amsmath}
\usepackage{fancyhdr}       
\usepackage{graphicx}       
\graphicspath{{media/}}     
\numberwithin{equation}{section}
\newtheorem{theorem}{Theorem}[section]

\newtheorem{corollary}[theorem]{Corollary}

\newtheorem{definition}[theorem]{Definition}

\newtheorem{proposition}[theorem]{Proposition}
\newtheorem{remark}[theorem]{Remark}

\newcommand{\beq} {\begin{equation}}
\newcommand{\eeq} {\end{equation}}
\newcommand{\df}{\textrm{d}}
\newcommand{\Df}{\textrm{D}}
  
\title{Schrödinger Connections: From Mathematical Foundations Towards Yano-Schrödinger Cosmology

}

\author{
  Lehel Csillag \thanks{Author to whom any correspondence should be addressed}\\
  Department of Physics,  
Babeș-Bolyai University  \\
Kogalniceanu Street, Cluj Napoca 400084, Romania and \\
Department of Mathematics and Computer Science,\\
Transilvania University of Brasov, Brașov, Romania \\
\texttt{lehel.csillag@ubbcluj.ro, lehel.csillag@unitbv.ro, lehel@csillag.ro}
\And
Anish Agashe \\
Department of Physics and Materials Science,
St. Mary's College of Maryland \\
47645 College Dr, St. Mary's City, Maryland, USA 20686 \\
\texttt{anagashe@smcm.edu} \\
  \And
 Damianos Iosifidis \\
Laboratory of Theoretical Physics, Institute of Physics,
University of Tartu \\
W. Ostwaldi 1, 50411 Tartu, Estonia \\
\texttt{damianos.iosifidis@ut.ee} \\
}

\begin{document}
\maketitle

\begin{abstract}
Schr\"odinger connections are a special class of affine connections, which despite being metric incompatible, preserve length of vectors under autoparallel transport. In the present paper, we introduce a novel coordinate-free formulation of Schrödinger connections. After recasting their basic properties in the language of differential geometry, we show that Schrödinger connections can be realized through torsion, non-metricity, or both. We then calculate the curvature tensors of Yano-Schrödinger geometry and present the first explicit example of a non-static Einstein manifold with torsion. We generalize the Raychaudhuri and Sachs equations to the Schr\"odinger geometry. The length-preserving property of these connections enables us to construct a Lagrangian formulation of the Sachs equation. We also obtain an equation for cosmological distances. After this geometric analysis, we build gravitational theories based on Yano-Schrödinger geometry, using both a metric and a metric-affine approach. For the latter, we introduce a novel cosmological hyperfluid that will source the Schr\"odinger geometry. Finally, we construct simple cosmological models within these theories and compare our results with observational data as well as the $\Lambda$CDM model. 

\end{abstract}

\newpage
 \tableofcontents

\section{Introduction}
The advancement of early 20th-century theoretical physics was closely intertwined with developments in mathematics. General relativity (GR) exemplifies this, being a product of collaborative efforts between mathematicians and physicists. The theory presented an elegant geometric perspective on gravity to the physics community, while simultaneously opening up new research domains in geometric analysis for mathematicians.
 
During the early days of GR, mathematicians identified two key mathematical premises of Einstein's theory: metric compatibility and torsion-freeness of the affine connection. These conditions uniquely define the Levi-Civita connection in terms of the metric. This realization led mathematicians to propose extensions to Einstein's theory, by relaxing either or both the conditions on the connection. A geometry equipped with such a general connection is known as metric-affine (or sometimes non-Riemannian) geometry \cite{Hehl:1994ue}.

Theories of gravity based on such geometries constitute promising and well motivated extensions of GR. The earliest instances of such extensions could be traced back to Weyl and Cartan merely a few years after Einstein proposed his theory \cite{Weyl,Cartan1,Cartan2,Cartan3,Cartan4}.
{It is also important to note that the rich geometric structure that comes from torsion and non-metricity has its origin in the microproperties of matter. Indeed, torsion is known to be sourced by the spin of matter, and non-metricity is sourced by the hadronic properties of matter associated to the so-called  shear and dilation currents of hypermomentum. We therefore see that such geometric modifications have their origin in the properties of the physical matter.  With many subcases possible due to inclusion (or exclusion) of different geometric features, there exists a whole landscape of non-Riemannian extensions of GR (see Figure \ref{fig-nrlandscape}). 

\begin{figure}[h]
    \centering
    \includegraphics[width=0.7\linewidth]{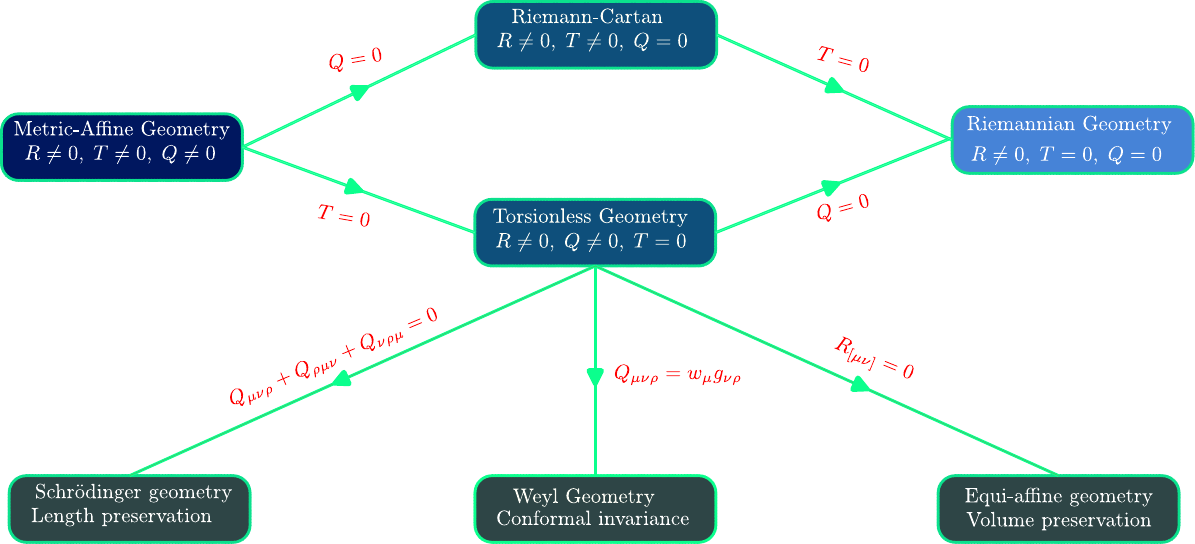}
    \caption{The landscape of non-Riemannian extensions of general relativity.} \label{fig-nrlandscape}
\end{figure}

These extensions have also been considered in the context of unified field theories \cite{unified}. The impact of non-Riemannian quantities on cosmology has also been of great interest (see \cite{Puetzfeld:2004yg} and the references within). Recently, torsional $f(T)$ modifications of gravity were used to address the Hubble tension, as pointed out in \cite{mandal2023h0,wang2020ft}. It is important to note, howevever, that $f(T)$ theories suffer from a strong coupling problem \cite{PhysRevD.103.024054}. For a detailed review on $f(T)$ theories consult 
\cite{cai2016ft}. Similarly, cosmological studies in theories containing non-metricity have also been done \cite{jim,khy}. Motivated by this, we analyse extensions of GR based on a particular type of non-metric connection: the Schr\"odinger connection. We detail the context behind using this connection below.

In 1924, Friedmann introduced a connection with vectorial torsion, which he called the semi-symmetric connection \cite{FriedmannSchouten}. Since then, mathematical aspects of semi-symmetric connections have been thoroughly examined \cite{Barua,Biswas,Imai,Agashe,Amur,Sharfudin,Meraj}. These investigations were propelled by the foundational work of Yano \cite{Kentaroyano}, in which a geometric, coordinate-free description of semi-symmetric metric connections was presented. This connection remained mostly unnoticed by physicists, until recently \cite{unificationsemisym, zangiabadi, Riccicalculusbook}. The reason perhaps being that the semi-symmetric connections lack a main physical requirement: they do not preserve lengths of parallely transported curves, unless they are metric-compatible. The same issue was also present in other extended gravity theories, most notably in Weyl's attempt at a unified field theory, as pointed out by Einstein \cite{Wheeler}. 



Schrödinger found the conditions that have to be imposed on a general affine connection in order to ensure the preservation of lengths during autoparallel transport \cite{Schrodinger}. In a manner similar to the semi-symmetric connection, Schrödinger's connection has also been overlooked in both mathematical and physics research, despite its remarkable property of being able to overcome Einstein's objections. Recently, it was reassessed from a physics standpoint in \cite{SilkeKlemm}, where a three-dimensional metric affine gravity theory was introduced, whose field equations give rise to a Schrödinger connection. Physically realistic cosmological models within a Weyl-Schrödinger type geometry were investigated in \cite{HarkoSchr}. Therein, the field equations were derived from a Lagrangian of the form $L=R+\frac{5}{24}Q_\rho Q^\rho+\frac{1} {6}\tilde{Q}_\rho\tilde{Q}^\rho+2T_\rho Q^\rho
+\zeta^{\rho\sigma}_{~~\alpha}T^\alpha_{~\rho\sigma}$, where $Q_\rho$, $\tilde{Q}_\mu$, and $T_\nu$ are the nonmetricity, torsion vectors, and $\zeta^{\rho\sigma}_{~~\alpha}$ are Lagrange multipliers, respectively. Note that hypermomentum was not considered in \cite{HarkoSchr}.

With this renewed interest in semi-symmetric connections, in this paper, we present a comprehensive investigation into the Schr\"odinger connections. This is of particular physical interest since a geometry equipped with a Schr\"odinger connection (dubbed as Schr\"odinger geometry) allows for fixed-length vectors. Our aim is to present results which would be of interest to both mathematicians and physicists. For the mathematical community, we introduce Schrödinger connections in a coordinate-free manner, providing a basis for further studies. For the physicists, we study the cosmological implications of two types of gravitational theories within Schr\"odinger geometry: one in the metric formalism and another in the Palatini formalism. We would like to note that the sections for physicists can be read almost independently from the formal mathematical parts.

This paper is organized as follows: section \ref{section2} provides a novel, geometric and coordinate-free exposition of Schrödinger connections. Within this section, in \ref{geometric} some physically remarkable properties of Schrödinger connections are presented. In subsection \ref{realizations} the focus shifts to providing examples of Schrödinger connections, culminating in theorem \ref{realizationtheorem}. Section \ref{section3} is dedicated to expressing Schrödinger connections in terms more familiar to the physics community. We start by showing that our general results from section \ref{section2} reproduce some formulae used in the physics literature, then move on to compute the curvature tensors of the Yano-Schrödinger geometry.  In subsection \ref{Einsteinmanifolds}, we study Einstein manifolds with vectorial, or semi-symmetric torsion, and the analogue of Schrödinger-Einstein manifolds. Through theorem \ref{Einsteinmanifoldtheorem}, we provide, for the first time in the literature, an explicit example of an Einstein manifold with torsion and a non-static metric. We conclude section \ref{section3} by presenting the kinematics of curves within the Schrödinger geometry, together with a novel Lagrangian formulation of the Sachs optical equation in modified gravity.

Section \ref{section4} is devoted to showing the physical applicability of Yano-Schrödinger connections. We propose extensions of general relativity by considering the Einstein field equations in Yano-Schrödinger geometry in two approaches: metric and Palatini. For the latter, we develop a novel type of hyperfluid, which sources the desired non-metricity. In both the cases, we consider the dynamics of a spatially flat Friedman-Lema\^itre-Robertson-Walker (FLRW) universe. We derive the generalized Friedmann equations and interpret the extra terms appearing in them as a form of geometric dark energy. We validate this interpretation through comparison with observational data. We reformulate the Friedmann equations in the redshift representation and examine a model where both ordinary matter and the additional terms are conserved. We compare our model's theoretical predictions with both observational data and the $\Lambda$CDM model.

A summary of results, outlook, and directions for further research are provided in section \ref{section5}. In appendix \ref{appendixA} and \ref{appendixB} some algebraic details regarding curvature tensors and the derivation of the Friedmann equations are presented. For the Palatini approach, the detailed calculations are given in \ref{appendixC} and \ref{appendixD}.

\section{Mathematical Foundations of Schrödinger Geometry}\label{section2}
Historically, Schrödinger's goal was to find the most general affine connection, $\nabla$, which preserves the length of vectors that are autoparallelly transported. In his book \cite{Schrodinger}, he concluded that this connection should be represented by the following connection coefficients
\begin{equation}
    \tensor{\Gamma}{^\lambda _\mu _\nu}=\tensor{\gamma}{^\lambda _\mu _\nu}+g^{\lambda \rho} S_{\rho \mu \nu},
\end{equation}
where $\tensor{\gamma}{^\lambda_\mu_\nu}$ denotes the coefficients of the Levi-Civita connection and $S_{\rho \mu \nu}$ is a tensor satisfying
\begin{equation}\label{Schrodingerconsideration}
    S_{\rho \mu \nu}=S_{\rho \nu \mu}, \; \; S_{(\rho \mu \nu)}=0.
\end{equation}
For contemporary mathematicians, reading Schrödinger's work might be challenging due to its reliance on local coordinates and outdated terminology. To bridge this gap, we reinterpret Schrödinger's concepts using a coordinate-free, geometric approach. Let's begin with the basic definitions.
\subsection{Geometric Definition and Physically Remarkable Properties}\label{geometric}
\begin{definition}
    Let $(M,g)$ be a semi-Riemannian manifold and the Levi-Civita connection be denoted with $\overset{\circ}{\nabla}$. An affine connection $\nabla$ on $M$ is called a \textbf{Schrödinger connection} if
    \begin{equation}\label{schrfree}
        \nabla_{X} Y=\overset{\circ}{\nabla}_{X} Y + U( -, X,Y) \; \; \forall X,Y \in \Gamma(TM),
    \end{equation}
    where $U$ is a $(1,2)$ tensor-field satisfying the following two conditions for all vector fields $X,Y$ and one-forms $\omega$:
    \begin{enumerate}
        \item[$(a)$] symmetry in last two entries
        \begin{equation}
            U(\omega,X,Y)=U(\omega,Y,X) ;
        \end{equation}
        \item[$(b)$] cyclicity
        \begin{equation}
        \begin{aligned}
U(\omega,X,Y)+U \left(\omega,Y,X\right) + U(X^{\flat},\omega^{\sharp},Y)
+U(X^{\flat},Y,\omega^{\sharp})+ U(Y^{\flat},\omega^{\sharp},X)+U(Y^{\flat},X,\omega^{\sharp})=0,
\end{aligned}
        \end{equation}
    where $X^{\flat}=g(X,-)$ and $\omega^{\sharp}=g^{-1}(\omega,-)$ are the musical isomorphisms.
    \end{enumerate}
\end{definition}
\begin{remark}
    The cyclicity  condition can be simplified by using symmetry. In particular, symmetry implies that
    \begin{equation}
        U(-,X,Y)=U(-,Y,X) \; \; \forall X,Y \in \Gamma(TM).
    \end{equation}
    Hence, cyclicity takes the form
    \begin{equation}
U(\omega,X,Y)+U(X^\flat,\omega^\sharp,Y)+U(Y^\flat,\omega^\sharp,X)=0.
    \end{equation}
\end{remark}
To relate with Schrödinger's work, we formally introduce the Schrödinger tensor, whose components coincide with \eqref{Schrodingerconsideration}.
\begin{definition}\label{definitionschrodingertensor}
    Let $(M,g)$ be a semi-Riemannian manifold with a given Schrödinger connection $(\nabla,U)$. The \textbf{Schrödinger tensor} $S$ associated to the Schrödinger connection $(\nabla,U)$ is defined as
    \begin{equation}\label{schrtensor}
    S:\Gamma(TM) \times \Gamma(TM) \times \Gamma(TM) \to C^{\infty}(M), \; \; S(A,X,Y):=U(A^{\flat},X,Y).
\end{equation}
\end{definition}
\begin{proposition}\label{Schrodingersymmetry}
    Let $\nabla$ be an affine connection on a semi-Riemannian manifold $(M,g)$ of the form
    \begin{equation}
        \nabla_{X} Y =\overset{\circ}{\nabla}_X Y+ U(-,X,Y),
    \end{equation}
    where $U(-,X,Y)$ is a $(1,2)$-tensor field and $X,Y$. Then the following are equivalent:
    \begin{enumerate}
        \item[$(i)$] The pair $(U,\nabla)$ is a Schrödinger connection;
        \item[$(ii)$] For all vector fields $A,X,Y$ the Schrödinger tensor $S$ associated to $U$ satisfies
        \begin{equation*}
            S(A,X,Y)=S(A,Y,X), \; \; S(A,X,Y)+S(Y,A,X)+S(X,Y,A)=0 .
        \end{equation*}
    \end{enumerate}
\end{proposition}
Schrödinger connections have two remarkable properties:
\begin{enumerate}
    \item[(S$1$)] They are torsion-free;
    \item[(S$2$)] The length of autoparallelly transported vectors does not change.
\end{enumerate}
We begin by establishing property (S$2$), as it is this characteristic that renders Schrödinger connections physically significant, making them overcome Einstein's concerns.  First, we prove a more general proposition, from which (S$2$) can be directly obtained.
\begin{proposition}
   Let $\nabla$ be an affine connection on a semi-Riemannian manifold $(M,g)$  of the form
    \begin{equation}\label{definingequation}
        \nabla_{X} Y= \overset{\circ}{\nabla}_X Y + U(-,X,Y)
    \end{equation}
     for vector fields $X,Y$ and a $(1,2)$-tensor field $U$. For all vector fields $X$, which satisfy $\nabla_{X}X=0$, the following are equivalent:
     \begin{enumerate}
         \item[$(i)$] $\nabla_{X}(g(X,X))=0$,
         \item[$(ii)$] $U\left( X^{\flat},X,X \right)=0$.
     \end{enumerate}
\end{proposition}
\begin{proof}
     Fix an arbitrary vector field $A \in \Gamma(TM)$ and suppose it is the tangent to a $\nabla$-autoparallel curve, so that $\nabla_{A}A=0$. Observe that substituting $X=Y=A$ in equation \eqref{definingequation}, we get
    \begin{equation}\label{usethis}
    \overset{\circ}{\nabla}_{A} A + U(-,A,A)=0.
    \end{equation}
    Since $\overset{\circ}{\nabla}$ is the Levi-Civita connection, we have $\overset{\circ}{\nabla}g=0$, so
    \begin{equation*}
        \begin{aligned}
            \nabla_{A}(g(A,A))&=\overset{\circ}{\nabla}_{A}(g(A,A)) =\overset{\circ}{\nabla}_{A} g(A,A)+ g \left( \overset{\circ}{\nabla}_{A} A,A\right) + g \left(A,\overset{\circ}{\nabla}_{A} A \right)\\
            &=2g \left( \overset{\circ}{\nabla}_{A} A,A \right)
            =-2g(U(-,A,A),A) =-2U\left( A^{\flat},A,A \right),
        \end{aligned}
    \end{equation*}
    where in the last step we used equation \eqref{usethis}. Hence
    \begin{equation}
        \nabla_{A}(g(A,A))=-2U \left(  A^{\flat},A,A\right).
    \end{equation}
    Therefore the condition that the parallel transport of $\nabla$ preserves lengths is equivalent to
    \begin{equation*}
       U \left( A^{\flat},A,A \right)=0, \; \; \text{as desired}.
    \end{equation*}
\end{proof}
\begin{corollary}
    Schrödinger connections preserve the length of autoparallelly transported vectors.
\end{corollary}
\begin{proof}
    By definition, the $(1,2)$-tensor field $U$ of a Schrödinger connection satisfies
   \begin{equation}
       U(\omega,X,Y)+ U \left(X^{\flat},\omega^{\sharp},Y \right) +U\left(Y^\flat, \omega^{\sharp},X \right)=0
    \end{equation}
    for all $X,Y \in \Gamma(TM)$ and $\omega \in \Omega^{1}(M)$. Choose $X=Y=A, \omega=A^{\flat}$, so
    \begin{equation*}
        U\left(A^{\flat},A,A \right)+ U \left(A^{\flat},A,A \right) +U( A^{\flat},A,A)=0 \iff   U\left(A^{\flat},A,A \right)=0.
    \end{equation*}
\end{proof}
This remarkable property is illustrated on Figure \ref{remarkable}.
\begin{figure*}[htbp]
\centering
\includegraphics[width=0.490\linewidth]{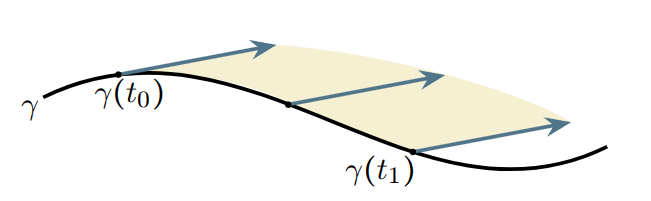} %
\includegraphics[width=0.490\linewidth]{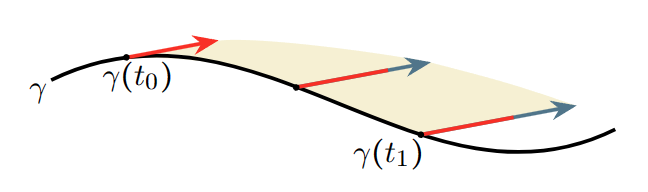}
\caption{ The left panel illustrates autoparallel transport of a Schrödinger connection, while the right panel depicts the autoparallel transport of a general affine connection.}
\label{remarkable}
\end{figure*}

In the following, we prove (S$1$).
\begin{proposition}
    Schrödinger connections are torsion-free.
\end{proposition}
\begin{proof}
    The torsion tensor of an affine connection $\nabla$ is defined as
    \begin{equation*}
        T(\omega,X,Y)=\omega(\nabla_{X} Y- \nabla_{Y}X - [X,Y]), \; \; \forall \omega \in \Omega^1(M), X,Y \in \Gamma(TM).
    \end{equation*}
    For  Schrödinger connections, one obtains
    \begin{equation*}
        T(\omega,X,Y)=\omega \left( \overset{\circ}{\nabla}_{X}Y +U(-,X,Y)- \overset{\circ}{\nabla}_{Y}X - U(-,Y,X)-[X,Y] \right),
    \end{equation*}
    which can be equivalently rewritten as
    \begin{equation}
        T(\omega,X,Y)=U(\omega,X,Y)-U(\omega,Y,X),
    \end{equation}
   since the Levi-Civita connection is torsion-free. Symmetry of  $U$ implies \begin{equation}
       T(\omega,X,Y)=0 \; \; \forall \omega \in \Omega^{1}(M), X,Y \in \Gamma(TM),
   \end{equation}
   thus concluding the proof.
\end{proof}
\begin{remark}
   In the modified gravity community, an affine connection is usually decomposed into two parts: one containing torsion, and one containing non-metricity. From the above proposition, we can see that Schrödinger connections, from the modified gravity point of view, contain only non-metricity. In this sense, the $U$-tensor (up to sign conventions) can be viewed as a special type of non-metricity, which admits fixed length vectors.
\end{remark}

\subsection{Realizations Using Torsion and Non-Metricity}\label{realizations}
After exploring the fundamental properties of Schrödinger connections, we move on to present particular instances of such connections. In \cite{SilkeKlemm}, it was shown that the Schr\"odinger connection can be realized through either torsion or non-metricity. Here, we clarify some aspects of such a realization, emphasizing that the said torsion must be of a \textit{distinct} affine connection. This is because the Schr\"odinger connection itself is, by construction, torsion-free as illustrated in the preceding proposition. The findings of this section can be summarized in the following theorem.

\begin{theorem}[Realization of Schrödinger Connections] \label{realizationtheorem} Let $\nabla$ be an affine connection  on a semi-Riemannian manifold $(M,g)$ of the form
\begin{equation}
    \nabla_{X}Y=\overset{\circ}{\nabla}_{X} Y + U(-,X,Y),
\end{equation}
    where $U(-,X,Y)$ is a $(1,2)$-tensor field. Fix a different affine connection $\overset{A}{\nabla}$ with torsion tensor $\overset{A}{T}$,non-metricity $\overset{A}{Q}$ satisfying
    \begin{equation}
        \overset{A}{Q}(X,Y,Z)+\overset{A}{Q}(Z,X,Y)+\overset{A}{Q}(Y,Z,X)=0
    \end{equation}
    and let $\pi,W$ be two one-forms  on $M$. If $U(\omega,X,Y)$ is given by any of the following list
    \begin{enumerate}
        \item[$(i)$] $ U(\omega,X,Y)=\frac{1}{2} \left( \overset{A}{T}\left(X^{\flat},Y,\omega^{\sharp} \right)+ \overset{A}{T} \left(Y^{\flat},X,\omega^{\sharp} \right) \right)$,
        \item[$(ii)$] $ U(\omega,X,Y)=\frac{1}{2} \left(X^{\flat} \left( \pi \left( \omega^{\sharp} \right) Y-\pi(Y) \omega^{\sharp} \right) + Y^{\flat} \left(\pi \left( \omega^{\sharp}\right)Y -\pi(X) \omega^{\sharp}\right)\right)$,
        \item[$(iii)$] $U(\omega,X,Y)=-\overset{A}{Q} \left(\omega^{\sharp},X,Y \right)$,
        \item[$(iv)$] $U(\omega,X,Y)=\frac{1}{2} \left( \overset{A}{T}\left(X^{\flat},Y,\omega^{\sharp} \right)+ \overset{A}{T} \left(Y^{\flat},X,\omega^{\sharp} \right) \right)- \overset{A}{Q}\left( \omega^{\sharp},X,Y \right)$,
        \item[$(v)$] $U(\omega,X,Y)=\pi\left( \omega^{\sharp} \right) W(X) W(Y) - \frac{1}{2} \left( W \left( \omega^{\sharp} \right) W(X) \pi(Y) + W\left( \omega^{\sharp} \right) W(Y) \pi(X) \right)$
    \end{enumerate}
    then $\nabla$ is a Schrödinger connection.
\end{theorem}
\begin{remark}
    The aim of this theorem is to clarify a  claim made in \cite{SilkeKlemm}, where it is stated that a Schrödinger connection can be written in terms of torsion, nonmetricity or both. This statement is misleading since by its very geometric nature, torsion has no effect on lengths. Therefore whether a connection has torsion or not is totally irrelevant for classifying it as a length-preserving Schroedinger connection.  Strictly speaking, one has to consider the torsion of a \underline{different affine connection}, as per se by definition the Schrödinger tensor must be symmetric in its latter indices. This implies, as we have shown above, that it must be torsion-free.
\end{remark}
We will dedicate the rest of the section to prove each item in the above theorem in detail.
\begin{proposition}
    Let $(M,g)$ be a semi-Riemannian manifold with an affine connection $\overset{A}{\nabla}$, whose torsion $\overset{A}{T}$ is non-vanishing. Then, the pair $(\nabla,U)$ defined by
    \begin{equation}\label{condition1tobeshown}
        U(\omega,X,Y)=\frac{1}{2} \left( \overset{A}{T}\left(X^{\flat},Y,\omega^{\sharp} \right)+ \overset{A}{T} \left(Y^{\flat},X,\omega^{\sharp} \right) \right)
    \end{equation}
    is a Schrödinger connection.
\end{proposition}   
\begin{proof}
    To prove $(\nabla,U)$ is a Schrödinger connection, we have to verify both symmetry and cyclicity. We start with symmetry and compute
    \begin{equation}\label{condition1shown}
        U(\omega,Y,X)=\frac{1}{2} \left( \overset{A}{T} \left( Y^{\flat},X,\omega^{\sharp} \right)+ \overset{A}{T} \left(X^{\flat},Y,\omega^{\sharp} \right) \right).
    \end{equation}
    It is evident that the right hand sides of equations \eqref{condition1tobeshown} and \eqref{condition1shown} agree. Consequently, it follows that
    \begin{equation}
        U(\omega,Y,X)=U(\omega,X,Y),
    \end{equation}
    thereby establishing symmetry. For cyclicity, we have
   \begin{equation}\label{addthis1}
        U(\omega,X,Y)=\frac{1}{2} \left( \textcolor{red}{\overset{A}{T}\left(X^{\flat},Y,\omega^{\sharp} \right)}+ \textcolor{blue}{\overset{A}{T} \left(Y^{\flat},X,\omega^{\sharp} \right)} \right)
    \end{equation}
    \begin{equation}\label{addthis2}
        U\left(X^{\flat},\omega^{\sharp},Y\right)=\frac{1}{2} \left( \textcolor{PineGreen}{\overset{A}{T}\left(\omega,Y,X\right)} + \textcolor{blue}{\overset{A}{T}\left(Y^{\flat},\omega^{\sharp},X \right)} \right)
    \end{equation}
    \begin{equation}\label{addthis3}
        U\left(Y^{\flat},\omega^{\sharp},X \right)=\frac{1}{2} \left( \textcolor{PineGreen}{\overset{A}{T} \left(\omega,X,Y \right)} + \textcolor{red}{\overset{A}{T} \left(X^{\flat},\omega^{\sharp},Y \right)} \right).
    \end{equation}
   After summing equations \eqref{addthis1}, \eqref{addthis2}, \eqref{addthis3}, the right-hand side cancels out entirely, as coded in colors, thanks to the symmetries of the torsion tensor. Hence, the desired equality
    \begin{equation}
         U(\omega,X,Y)+ U\left(X^{\flat},\omega^{\sharp},Y\right)+ U\left(Y^{\flat},\omega^{\sharp},X\right)=0
    \end{equation}
    is obtained.
\end{proof}
\begin{corollary}
    The Schrödinger tensor of a Schrödinger connection $(\nabla,U)$ implemented by torsion $\overset{A}{T}$ of an affine connection is given by
    \begin{equation} \label{Schrodingertorsion}
    \begin{aligned}
       & S:\Gamma(TM) \times \Gamma(TM) \times \Gamma(TM) \to C^{\infty}(M)\\
        & S(A,X,Y)=\frac{1}{2} \left( \overset{A}{T}
         \left(Y^{\flat},X, A\right) + \overset{A}{T} \left( X^{\flat},Y,A \right)\right)
    \end{aligned}
    \end{equation}
\end{corollary}
In the introduction we mentioned that Kentaro Yano explored the mathematical properties of an interesting connection, with torsion of the semi-symmetric type. In the following, recalling this notion, we provide an example of a Schrödinger connection, which is realized this way.
\begin{definition}
    On a semi-Riemannian manifold $(M,g)$ a connection $\nabla$ is called \textbf{semi-symmetric} if there exists a one-form $\pi$ such that
    \begin{equation*}
        T(\omega,X,Y)=\omega(\pi(Y)X-\pi(X)Y) \; \; \forall X,Y \in \Gamma(TM), \omega \in \Omega^{1}(M).
    \end{equation*}
\end{definition}
\begin{corollary}
    Given a one-form $\pi$, the pair $(\nabla,U)$ defined by
    \begin{equation*}
        U(\omega,X,Y)=\frac{1}{2} \left(X^{\flat} \left( \pi \left( \omega^{\sharp} \right) Y-\pi(Y) \omega^{\sharp} \right) + Y^{\flat} \left(\pi \left( \omega^{\sharp}\right)Y -\pi(X) \omega^{\sharp}\right)\right)
    \end{equation*}
    is a Schrödinger connection, called the \textbf{Yano-Schrödinger} connection.
\end{corollary}
In the following, we will implement a Schrödinger connection using the non-metricity tensor, which is defined as 
\begin{equation}
    Q(X,Y,Z)=-(\nabla_{X}g)(Y,Z), \; \; \forall X,Y,Z \in \Gamma(TM).
\end{equation}
\begin{proposition}
Let $(M,g)$ be a semi-Riemannian manifold with an affine connection $\overset{A}{\nabla}$, whose non-metricity $\overset{A}{Q}$ is non-vanishing and satisfies
\begin{equation}\label{assumption}
    \overset{A}{Q}(X,Y,Z)+\overset{A}{Q}(Z,Y,X)+\overset{A}{Q}(X,Z,Y)=0.
\end{equation} 
Then, there exists a Schrödinger connection $(\nabla,U)$, which is implemented by $\overset{A}{Q}$.
\end{proposition}
\begin{proof}
   We define the Schrödinger connection $(\nabla,U)$ as
   \begin{equation}\label{cyclicity1}
       U(\omega,X,Y)=-\overset{A}{Q} \left(\omega^{\sharp},X,Y \right) \; \; \forall \omega \in \Omega^{1}(M),X,Y \in \Gamma(TM).
   \end{equation}
   To see that $U$ is symmetric in its last two entries, it suffices to show that the non-metricity is, which is a direct consequence of the Leibniz rule for the connection $\overset{A}{\nabla}$. For cyclicity, consider
   \begin{equation}\label{cyclicity2}
       U\left(Y^{\flat},\omega^{\sharp},X \right)=-\overset{A}{Q} \left(Y,\omega^{\sharp},X \right),
   \end{equation}
   \begin{equation}\label{cyclicity3}
       U \left(X^{\flat},\omega^{\sharp},Y \right)= -\overset{A}{Q} \left(X,\omega^{\sharp},Y \right)=-\overset{A}{Q}\left(X,Y,\omega^{\sharp} \right).
   \end{equation}
   We add together the three equations \eqref{cyclicity1}, \eqref{cyclicity2}, and \eqref{cyclicity3}. By assumption \eqref{assumption}, the right hand sides vanish, so
   \begin{equation}
         U(\omega,X,Y)+   U\left(Y^{\flat},\omega^{\sharp},X \right)+ U \left(X^{\flat},\omega^{\sharp},Y \right)=0,
   \end{equation}
   which is exactly the cyclicity condition.
\end{proof}
Finally, we come to prove the last point of the theorem.
\begin{proposition}
    Let  $U(-,X,Y)$ be a $(1,2)$-tensor field on a semi-Riemannian manifold $(M,g)$, given by
    \begin{equation}
       U(\omega,X,Y)=\pi\left( \omega^{\sharp} \right) W(X) W(Y) - \frac{1}{2} \left( W \left( \omega^{\sharp} \right) W(X) \pi(Y) + W\left( \omega^{\sharp} \right) W(Y) \pi(X) \right).
    \end{equation}
    Then, the pair $(\nabla,U)$ is a Schrödinger connection.
\end{proposition}
\begin{proof}
    For the pair $(\nabla,U)$ to be a Schrödinger connection, $U$ must be symmetric in its last two entries and it must satisfy cyclicity. As symmetry is clear by construction, we only verify the cyclicity condition. To this end, consider
    \begin{equation}\label{thisone}
        U(\omega,X,Y)=\textcolor{red}{\pi\left( \omega^{\sharp} \right) W(X) W(Y)} - \frac{1}{2} \left( \textcolor{blue}{W \left( \omega^{\sharp} \right) W(X) \pi(Y)} + \textcolor{orange}{W\left( \omega^{\sharp} \right) W(Y) \pi(X)} \right)
    \end{equation}
    \begin{equation}\label{thistwo}
        U\left(Y^{\flat},\omega^{\sharp},X \right)=\textcolor{blue}{\pi(Y) W\left( \omega ^{\sharp} \right) W(X)} - \frac{1}{2} \left(\textcolor{orange}{W(Y) W \left( \omega^{\sharp} \right) \pi(X)}+ \textcolor{red}{W(Y) W(X) \pi \left( \omega^{\sharp} \right)} \right)
    \end{equation}
    \begin{equation}\label{thisthree}
        U\left(X^{\flat},\omega^{\sharp},Y \right)=\textcolor{orange}{\pi(X) W\left( \omega^{\sharp} \right) W(Y)} - \frac{1}{2} \left( \textcolor{blue}{W(X) W \left(\omega^{\sharp} \right) \pi(Y)} + \textcolor{red}{W(X) W(Y) \pi \left( \omega^\sharp \right)} \right)
    \end{equation}
After summing equations \eqref{thisone}, \eqref{thistwo}, \eqref{thisthree}, the right hand side cancels out entirely, as coded in colors. Hence, the desired equality
\begin{equation}
    U(\omega,X,Y)+U \left(X^{\flat},\omega^{\sharp},Y \right) + U \left(Y^{\flat},\omega^{\sharp},X \right)=0
\end{equation}
is obtained. This shows that the pair $(\nabla,U)$ is a Schrödinger connection.
\end{proof}

\section{Local Description: A Physicist's Perspective}\label{section3}
In this section, in contrast to the earlier global approach, we choose local charts.  We calculate and present new coordinate expressions for the curvature tensors of Yano-Schrödinger geometry in Appendix \ref{coordcurv}.

\subsection{Einstein-Schrödinger Manifolds}\label{Einsteinmanifolds}
Einstein manifolds are of particular importance for physics, since they serve as solutions to the Einstein equations in vacuum. While these manifolds have been well-studied in (semi-)Riemannian geometry \cite{Besse}, they have only recently become a topic of interest in metric-affine theories.  In \cite{raveratorsionnonmetricity} Klemm and Ravera introduced the concept of an Einstein manifold with torsion and non-metricity and provided a characterization  in terms of partial differential equations in local coordinates.

A special case is when the non-metricity is of the Weyl type, leading to what's known as an Einstein-Weyl space. These spaces have been the subject of research \cite{lionelmason}, and recently, the challenge of finding an action principle for the Einstein-Weyl equations was resolved \cite{raveratorsionnonmetricity}.  Moreover, Einstein-Weyl spaces with torsion have also been explored \cite{mustafathesis}.

In this section we explore Einstein manifolds with semi-symmetric torsion and consequently Einstein-Schrödinger manifolds, where the Schrödinger connection is realized through a semi-symmetric type of torsion. We will present a concrete example of a non-static Einstein manifold with torsion, with the help of a semi-symmetric connection by solving the differential equations characterizing an Einstein manifold in a chart. To start, let us review the basic notion of an Einstein metric.

\begin{definition}
    A semi-Riemannian metric $g$ on a smooth manifold $M$ is called an \textbf{Einstein metric} if there exists a smooth function $\lambda:M \to \mathbb{R}$, such that
    \begin{equation}\label{einsteincondition}
        \overset{\circ}{R}_{\mu \nu}=\lambda g_{\mu \nu}.
    \end{equation}
\end{definition}
\begin{remark}
    Note that for a general affine connection $\nabla$ equation \eqref{einsteincondition} does not make sense, as $R_{\mu \nu}$ is not symmetric in general. Hence, we will symmetrize in our following definition accordingly.
\end{remark}
\begin{definition}
A semi-Riemannian manifold $(M,g)$ equipped with an affine connection $\nabla$ is called a \textbf{generalized Einstein manifold} if there exists a smooth function $\lambda:M \to \mathbb{R}$ such that
\begin{equation}
    R_{(\mu \nu)}=\lambda g_{\mu \nu}.
\end{equation}
In addition, the generalized Einstein manifold with $\nabla=(\nabla,U)$, a Yano-Schrödinger connection, will be called an \textbf{Einstein-Schrödinger} manifold.  
\end{definition}
\begin{remark}
    It should be noted that the above definition gives actually a constraint equation for the connections satisfying it.
\end{remark}
\begin{remark}
    Generalized Einstein manifolds need not be Einstein manifolds in the usual sense. Nevertheless, choosing $(\nabla,U)=\overset{\circ}{\nabla}$ to be the Levi-Civita connection, we obtain an Einstein manifold, upon choosing $\lambda=const$. In case we choose a different connection, but an Einstein metric, we get a partial differential equation in charts for the torsion and non-metricity part.
\end{remark}
We now turn our attention to a special case, namely that of an Einstein manifold with semi-symmetric type of torsion.
\begin{proposition}
    Let $(M,g)$ be a semi-Riemannian manifold equipped with a semi-symmetric metric connection. Then the following are equivalent:
    \begin{enumerate}
        \item[$(i)$] $(M,g)$ is a generalized Einstein-manifold.
        \item[$(ii)$] $ \overset{\circ}{R}_{\mu \nu} - \overset{\circ}{\nabla}_{\mu} \pi_{\nu} - \overset{\circ}{\nabla}_{\nu} \pi_{\mu}+ 2 \pi_{\nu} \pi_{\mu} +\frac{1}{2} g_{\mu \nu} \overset{\circ}{\nabla}_{\lambda} \pi^{\lambda} - \frac{1}{2} g_{\mu \nu} \pi_{\lambda} \pi^{\lambda}=\frac{1}{4} g_{\mu \nu} \overset{\circ}{R}$.
    \end{enumerate}
\end{proposition}
\begin{proof}
The symmetrized Ricci tensor is given by \cite{csillag2024semisymmetric}
\begin{equation}
    R_{(\mu \nu)}=\overset{\circ}{R}_{\mu \nu} -\overset{\circ}{\nabla}_\mu \pi_\nu - \overset{\circ}{\nabla}_\nu \pi_\mu +2 \pi_\nu \pi_\mu - g_{\mu \nu} \overset{\circ}{\nabla}_{\lambda} \pi^{\lambda} -2 g_{\mu \nu} \pi_{\lambda} \pi^{\lambda}.
\end{equation}
Therefore, by contracting the condition for $(M,g)$ to be a generalized Einstein manifold, we obtain
\begin{equation}
\begin{aligned}
    \overset{\circ}{R}- \overset{\circ}{\nabla}_{\mu} \pi^{\mu} - \overset{\circ}{\nabla}_{\nu} \pi^{\nu} + 2\pi_\nu \pi^\nu -4 \overset{\circ}{\nabla}_{\lambda} \pi^{\lambda} -8 \pi_{\lambda} \pi^{\lambda}= 4 \lambda,
\end{aligned}
\end{equation}
from which expressing $\lambda$ gives
\begin{equation}
    \lambda=\frac{1}{4} \left(\overset{\circ}{R} -6 \nabla_{\lambda} \pi^{\lambda}  -6 \pi_\lambda \pi^{\lambda} \right).
\end{equation}
Hence $(M,g)$ is a generalized Einstein manifold iff
\begin{equation}
    \overset{\circ}{R}_{\mu \nu} -\overset{\circ}{\nabla}_\mu \pi_\nu - \overset{\circ}{\nabla}_\nu \pi_\mu +2 \pi_\nu \pi_\mu - g_{\mu \nu} \overset{\circ}{\nabla}_{\lambda} \pi^{\lambda} -2 g_{\mu \nu} \pi_{\lambda} \pi^{\lambda}=\frac{1}{4} \left(\overset{\circ}{R} -6 \nabla_{\lambda} \pi^{\lambda}  -6 \pi_\lambda \pi^{\lambda} \right) g_{\mu \nu}.
\end{equation}
Rearranging leads to the desired result
\begin{equation}
    \overset{\circ}{R}_{\mu \nu} - \overset{\circ}{\nabla}_{\mu} \pi_{\nu} - \overset{\circ}{\nabla}_{\nu} \pi_{\mu}+ 2 \pi_{\nu} \pi_{\mu} +\frac{1}{2} g_{\mu \nu} \overset{\circ}{\nabla}_{\lambda} \pi^{\lambda} - \frac{1}{2} g_{\mu \nu} \pi_{\lambda} \pi^{\lambda}=\frac{1}{4} g_{\mu \nu} \overset{\circ}{R}.
\end{equation}
\end{proof}
\begin{corollary}
    A semi-Riemannian manifold $(M,g)$ equipped with a semi-symmetric connection and an Einstein metric is a generalized Einstein manifold iff
    \begin{equation} \label{piphi}
         -\overset{\circ}{\nabla}_{\mu} \pi_{\nu} - \overset{\circ}{\nabla}_{\nu} \pi_{\mu}+ 2 \pi_{\nu} \pi_{\mu} +\frac{1}{2} g_{\mu \nu} \overset{\circ}{\nabla}_{\lambda} \pi^{\lambda} - \frac{1}{2} g_{\mu \nu} \pi_{\lambda} \pi^{\lambda}=0.
    \end{equation}
\end{corollary}

We have a similar characterization of Einstein-Schrödinger manifolds in local charts.
\begin{proposition}\label{yanoeinstein}
    Let $(M,g)$ be a semi-Riemannian manifold equipped with a Yano-Schrödinger connection $(\nabla,U)$. Then the following are equivalent:
    \begin{enumerate}
        \item[$(i)$] $(M,g)$ is an Einstein-Schrödinger manifold.
        \item[$(ii)$] $
            \overset{\circ}{R}_{\mu \nu} -\frac{1}{8} g_{\mu \nu} \overset{\circ}{\nabla}_{\alpha} \pi^{\alpha} + \frac{1}{2} \overset{\circ}{\nabla}_{\mu} \pi_{\nu} + \frac{1}{2} \overset{\circ}{\nabla}_{\nu} \pi_{\mu} + \frac{1}{16} g_{\mu \nu} \pi^{\alpha} \pi_{\alpha} - \frac{1}{4} \pi_{\nu} \pi_{\mu}=\frac{1}{4} \overset{\circ}{R}  g_{\mu \nu}$.
    \end{enumerate}
\end{proposition}
\begin{proof}
    The symmetrization of the Ricci tensor of the Yano-Schrödinger connection is given by
    \begin{equation}
    \begin{aligned}
        R_{(\mu \nu)}
        =\overset{\circ}{R}_{\mu \nu}+g_{\mu \nu} \overset{\circ}{\nabla}_{\alpha} \pi^{\alpha} + \frac{1}{4} \overset{\circ}{\nabla}_{\mu} \pi_{\nu} + \frac{1}{4} \overset{\circ}{\nabla}_{\nu} \pi_{\mu} - \frac{1}{2} g_{\mu \nu} \pi^\alpha \pi_\alpha - \frac{1}{4} \pi_\nu \pi_\mu .
    \end{aligned}
    \end{equation}
    By contracting the condition to be an Einstein-Schrödinger manifold and expressing $\lambda$, we obtain
    \begin{equation}
        \lambda=\frac{1}{4} \left( \overset{\circ}{R}+ \frac{9}{2} \overset{\circ}{\nabla}_{\mu} \pi^{\mu} - \frac{9}{4} \pi^\alpha \pi_\alpha \right).
    \end{equation}
    Hence $(M,g)$ is Einstein iff
    \begin{equation}
        \overset{\circ}{R}_{\mu \nu}+g_{\mu \nu} \overset{\circ}{\nabla}_{\alpha} \pi^{\alpha} + \frac{1}{2} \overset{\circ}{\nabla}_{\mu} \pi_{\nu} + \frac{1}{2} \overset{\circ}{\nabla}_{\nu} \pi_{\mu} - \frac{1}{2} g_{\mu \nu} \pi^\alpha \pi_\alpha - \frac{1}{4} \pi_\nu \pi_\mu = \frac{1}{4} \left( \overset{\circ}{R}+ \frac{9}{2}\overset{\circ}{\nabla}_{\alpha} \pi^{\alpha} - \frac{9}{4} \pi^\alpha \pi_\alpha \right) g_{\mu \nu}.
    \end{equation} 
    Rearranging leads to the desired result.
\end{proof}
\begin{corollary}
    A semi-Riemannian manifold $(M,g)$ equipped with an Einstein metric and a Yano-Schrödinger connection $(\nabla,U)$ is an Einstein-Schrödinger manifold iff
    \begin{equation}
        -\frac{1}{8} g_{\mu \nu} \overset{\circ}{\nabla}_{\alpha} \pi^{\alpha} + \frac{1}{2} \overset{\circ}{\nabla}_{\mu} \pi_{\nu} + \frac{1}{2} \overset{\circ}{\nabla}_{\nu} \pi_{\mu} + \frac{1}{16} g_{\mu \nu} \pi^{\alpha} \pi_{\alpha} - \frac{1}{4} \pi_{\nu} \pi_{\mu}=0.
    \end{equation}
\end{corollary}
Tracing the latter it follows that
\beq
(n-4)\Big[-2\overset{\circ}{\nabla}_{\alpha}\pi^{\alpha}+\pi_{\alpha}\pi^{\alpha} \Big]=0.
\eeq
Note that this is trivialized for $n=4$.

We conclude the section with a surprising non-static generalized Einstein manifold.
\begin{theorem}\label{Einsteinmanifoldtheorem}
    There exists a generalized Einstein manifold $(M,g)$, where $g$ is not static.
\end{theorem}
\begin{proof}
    The proof is constructive. We start by choosing $M=\mathbb{R}^{4}$ and equip with with the metric
    \begin{equation}
        ds^2=-dt^2+a(t)^2 \delta_{ij} dx^i dx^j,
    \end{equation}
    where we assume that $\frac{\dot a}{a}=H_0$ is a non-zero constant. Moreover, we equip $\left(\mathbb{R}^{4},ds^2 \right)$ with a semi-symmetric connection. According to proposition , the considered manifold is a generalized Einstein manifold iff the equation
    \begin{equation}
          \overset{\circ}{R}_{\mu \nu} - \overset{\circ}{\nabla}_{\mu} \pi_{\nu} - \overset{\circ}{\nabla}_{\nu} \pi_{\mu}+ 2 \pi_{\nu} \pi_{\mu} +\frac{1}{2} g_{\mu \nu} \overset{\circ}{\nabla}_{\lambda} \pi^{\lambda} - \frac{1}{2} g_{\mu \nu} \pi_{\lambda} \pi^{\lambda}=\frac{1}{4} g_{\mu \nu} \overset{\circ}{R}
    \end{equation}
    is satisfied. As the metric possesses high symmetry, we choose
    \begin{equation}
        \pi_{\mu}=(\psi(t),0,0,0) \iff \pi^{\mu}=(-\psi(t),0,0,0).
    \end{equation}
    Hence, we have to satisfy the following system of differential equations for $\psi(t)$
    \begin{equation}
        -3\frac{\ddot a}{a} -\dot \psi - \dot \psi +2 \psi^2-\frac{1}{2} \left(- \dot \psi -3 \frac{\dot a}{a} \psi \right) - \frac{1}{2}\psi^2=-\frac{6}{4} \left( \frac{\ddot a}{a}+\frac{\dot a^2}{a^2} \right),
    \end{equation}
    \begin{equation}
        a \ddot a +2 \dot a^2 +2 a \dot a \psi +\frac{1}{2}a^2 \left(- \dot \psi - 3 \frac{\dot a}{a}\psi \right) +\frac{1}{2} a^2 \psi^2=\frac{6}{4} a^2\left( \frac{\ddot a}{a}+\frac{\dot a^2}{a^2} \right).
    \end{equation}
    Rearranging, algebraically simplifying, we get
    \begin{equation}
        -\frac{3}{2} \frac{ \ddot a}{a}+\frac{3}{2} \frac{\dot a^2}{a^2} -\frac{3}{2} \dot \psi + \frac{3}{2}\psi^2 +\frac{3}{2} \frac{\dot a}{a} \psi=0,
    \end{equation}
    \begin{equation}
        -\frac{1}{2} \frac{\ddot a}{a} + \frac{1}{2} \frac{\dot a^2}{a^2} + \frac{1}{2} \frac{\dot a}{a} \psi - \frac{1}{2} \dot \psi +\frac{1}{2} \psi^2=0.
    \end{equation}
    By introducing the Hubble parameter $H_0=\frac{\dot a}{a}$, using the assumption that $H_0$ is constant, we obtain
 \begin{equation}
        -\frac{3}{2} \dot \psi + \frac{3}{2} \psi^2 + \frac{3}{2} H_0 \psi=0,
    \end{equation}
    \begin{equation}
        \frac{1}{2} H_0 \psi - \frac{1}{2} \dot \psi + \frac{1}{2} \psi^2=0.
    \end{equation}
    We can see that the two equations are identical. Hence, we can solve for example the first one, i.e.
    \begin{equation}
        \dot \psi - \psi^2 - H_0 \psi=0,
    \end{equation}
   which is a Bernoulli type differential equation. The solution is given by 
    \begin{equation}
        \psi(t)=-\frac{H_0 e^{H_0(t_0+t)}}{e^{H_0(t_0+t)}-1},
    \end{equation}
    where $t_0$ is an arbitrary integration constant fixed by the initial condition $\psi(0)$.
\end{proof}
\begin{corollary}
    There exists a generalized Einstein manifold, whose metric $g$ is not an Einstein metric.
\end{corollary}
\subsection{Kinematics of Curves in Schrödinger Geometry}
In this section, we derive the equations governing the kinematics of congruence of curves on a manifold equipped with a Schrödinger connection. These equations describe the evolution of the irreducible components of the transverse part of the velocity gradient, `velocity' being the tangent vector to the congruence. The equations are essentially geometric identities and thus, they are kinematic in nature: no dynamics are assumed in deriving them. The Raychaudhuri equation (RE) for timelike congruences in metric-affine gravity was first obtained in \cite{Iosifidis_2018}, while further investigations were done in \cite{agakine}. In the present paper, we will derive these equations both for timelike and null congruences in Schrödinger geometry. For null congruences in Riemannian geometry, the equation was first obtained by Sachs \cite{sachs_1961}, hence we will call it the Sachs equation. Finally, we will extend the Lagrangian formulation proposed in \cite{Agashe_2024} to the case of null congruences.

\subsubsection{The Raychaudhuri Equation}
From now on, we adopt the modified gravity notation for the $U$-tensor, as we will discuss physical implications. In this setting, a Schrödinger connection corresponds to an affine connection with zero torsion and non-metricity satisfying
\begin{equation}
    Q_{(\mu \nu \rho)}=0.
\end{equation}
Let $u^{\alpha}$ be a tangent vector field to a congruence of timelike curves on a semi-Riemannian manifold equipped with a Schrödinger connection. We take $u^{\alpha}$ to be normalized
\begin{equation}
    u^\alpha u_\alpha = -1.
\end{equation}
By treating the tangent vector to be of fixed length, we are restricting our calculations to congruences of autoparallel curves with respect to either the Schr\"odinger connection or the Levi-Civita connection.

Since we are interested in the transverse properties of the congruence, we introduce the transverse metric
\begin{equation}
    h_{\alpha\beta} = g_{\alpha\beta} + u_\alpha u_\beta.
\end{equation}
It is easy to check that $h_{\alpha \beta}$ satisfies  $ u^{\alpha} h_{\alpha \beta}=0=h_{\alpha \beta} u^{\beta} $ and $g^{\alpha \beta} h_{\alpha \beta}=3$. With the help of $h_{\alpha \beta}$ we can separate the velocity gradient $\nabla_{\alpha} u^{\beta}$ into transverse and longitudinal parts
\begin{equation}\label{transvelgrad}
    ^{{\rm (T)}}\nabla_\alpha u^\beta = {h^\rho}_\alpha \tensor{h}{^\beta _\epsilon} \nabla_\rho u^\epsilon, \; \;
    ^{{\rm (L)}}\nabla_\alpha u^\beta= \nabla_\alpha u^\beta - ^{{\rm (T)}}\nabla_\alpha u^\beta .
\end{equation}
Using the definition of the transverse metric, one obtains
\begin{equation}
    ^{{\rm (T)}}\nabla_\alpha u^\beta = \nabla_\alpha u^\beta + u_\alpha A^\beta - Q_{\rho\epsilon\alpha}u^\rho u^\epsilon u^\beta,
\end{equation}
where $A^\beta := u^\alpha\nabla_\alpha u^\beta$ is the path-acceleration of the curves. For curves that are autoparallel with respect to the Schr\"odinger connection, the path-acceleration vanishes by definition. On the other hand, for the autoparallel curves with respect to the Levi-Civita connection, we have, $A^\beta = -{Q^\beta}_{\alpha\sigma}u^\alpha u^\sigma$. Since both these curves are geometrically/physically interesting, we will keep the path-acceleration term as it is in the calculations that follow. Then, the kinematics of a particular kind of congruence could be found by replacing the path-acceleration to be either of the two choices. 

The velocity gradient is decomposed into its irreducible components in the following manner
\begin{equation} \label{kinedef}
    \theta := {h^\alpha}_\beta \nabla_\alpha u^\beta , \quad
    {\omega_\alpha}^\beta := {h^\rho}_{[\alpha}{h^{\beta]}}_\epsilon \nabla_\rho u^\epsilon, \quad
    {\sigma_\alpha}^\beta := {h^\rho}_{(\alpha}{h^{\beta)}}_\epsilon \nabla_\rho u^\epsilon - \frac{1}{3} \theta {h_\alpha}^\beta,
\end{equation}
where $\theta$ is the trace part and is called expansion, ${\omega_\alpha}^\beta$ is the anti-symmetric part and is called rotation or vorticity, and ${\sigma_\alpha}^\beta$ is the symmetric-traceless part and is called shear. Using these, one can write

\begin{equation}\label{velgrad}
    \nabla_\alpha u^\beta = \frac{1}{3} \theta {h_\alpha}^\beta + {\omega_\alpha}^\beta + {\sigma_\alpha}^\beta - u_\alpha A^\beta + Q_{\rho\epsilon\alpha}u^\rho u^\epsilon u^\beta.
\end{equation}

The kinematic equations are simply the evolution equations of the three quantities defined above. In the following, we will derive (in detail) the evolution equation for the expansion $\theta$, which is also called the Raychaudhuri equation. The RE (and its null counterpart) is extremely important in relativistic cosmology and many other physical contexts \cite{ellisbook,schneidbook,hawkell,jacob,kar}. Using the definition of the expansion scalar, we can write
\begin{equation}
    \frac{D\theta}{d\lambda} = u^\rho\nabla_\rho \theta =  u^\rho\nabla_\rho \left( {h^\alpha}_\beta \nabla_\alpha u^\beta \right) = \delta^\alpha_\beta u^\rho\nabla_\rho \nabla_\alpha u^\beta ,
\end{equation}
where $\lambda$ is some parameter that parametrizes the curves, and the second equality is obtained by using 
\begin{equation}
    u^\rho A_\rho=u_\rho A^\rho=0, \; \; \text{and} \; \; \nabla_{\rho} \delta^{\alpha}_\beta=0.
\end{equation}
To derive the Raychaudhuri equation, we employ the Ricci identity
\begin{equation}\label{ricciid}
    \nabla_{[\rho} \nabla_{\alpha]} u^\beta = \frac{1}{2} {R^\beta}_{\epsilon\rho\alpha}u^\epsilon.
\end{equation}
Using this, the above equation becomes
\begin{equation}
    \frac{\Df\theta}{\df\lambda} = -R_{\alpha\beta}u^\alpha u^\beta + \nabla_\alpha A^\alpha - \nabla_\alpha u^\beta \nabla_\beta u^\alpha.
\end{equation}
Using equation \eqref{velgrad}, it can be shown that
\begin{equation}
    \nabla_\alpha u^\beta \nabla_\beta u^\alpha = \frac{1}{3}\theta^2 - \omega^2 + \sigma^2 - Q_{\alpha\beta\epsilon}A^\alpha u^\beta u^\epsilon.
\end{equation}
where, we have defined $\omega^2 := {\omega_\alpha}^\beta {\omega^\alpha}_\beta$ and $\sigma^2 := {\sigma_\alpha}^\beta {\sigma^\alpha}_\beta$. This simplifies the previous equation to give
\begin{equation}\label{rayeq}
    \frac{\Df\theta}{\df\lambda} = -R_{\alpha\beta}u^\alpha u^\beta - \frac{1}{3}\theta^2 + \omega^2 - \sigma^2  + \nabla_\alpha A^\alpha + Q_{\alpha\beta\epsilon}A^\alpha u^\beta u^\epsilon,
\end{equation}
which is the famous Raychaudhuri equation. Note that this is essentially a geometric identity. However, it can be converted into a dynamical equation when used in conjunction with (assumed) field equations. In the absence of the non-metricity, it reduces to the usual RE in Riemannian geometry.

\subsubsection{The Sachs Optical Equation}
We employ the same geometric setup as in the previous subsection to derive the null counterpart of the RE. Let $k^\alpha$ be a tangent vector field to a congruence of null curves on a semi-Riemannian manifold $(M,g)$ equipped with a Schrödinger connection characterized by the non-metricity $\boldsymbol{Q}$. As $k^\alpha$ is tangent to a congruence of null curves, we have
\begin{equation}
    k^\alpha k_\alpha = 0.
\end{equation}
In this case, the transverse metric is given by
\begin{equation}
    h_{\alpha\beta} = g_{\alpha\beta} + k_\alpha n_\beta + n_\alpha k_\beta,
\end{equation}
where $n_\alpha$ is an auxiliary null vector field, such that $n^\alpha n_\alpha = 0$ and $n^\alpha k_\alpha = -1$. Once again, it is easy to check that $k^\alpha h_{\alpha\beta} = 0 =h_{\alpha\beta}u^\beta$, $n^\alpha h_{\alpha\beta} = 0 =h_{\alpha\beta}n^\beta$, and $g^{\alpha\beta} h_{\alpha\beta} = 2$.

We note again that the calculations presented here are valid for autoparallel curves with respect to either Schr\"odinger or Levi-Civita connection. This is because only for these two types of curves, the tangent vectors have fixed lengths. Arbitrary null curves may not remain null when parallel transported and thus deriving the kinematic equations will become much more complicated. The same is true in the presence of an arbitrary non-metricity as well.

Using the transverse metric, the velocity gradient can be separated into transverse and longitudinal parts in the same manner as equation \eqref{transvelgrad}. Furthermore, we can write the irreducible components of the transverse part as
\begin{equation} \label{nullexp}
    \Theta := {h^\alpha}_\beta \nabla_\alpha k^\beta ,  \quad
    {\Omega_\alpha}^\beta := {h^\rho}_{[\alpha}{h^{\beta]}}_\epsilon \nabla_\rho k^\epsilon, \quad
    {\Sigma_\alpha}^\beta := {h^\rho}_{(\alpha}{h^{\beta)}}_\epsilon \nabla_\rho k^\epsilon - \frac{1}{2} \Theta {h_\alpha}^\beta,
\end{equation}
where $\Theta,\ \Omega$, and $\Sigma$ are called the expansion, rotation, and shear, respectively. With the help of these, one can write
\begin{equation}\label{velgrad1}
    \nabla_\alpha k^\beta = \frac{1}{2} \Theta {h_\alpha}^\beta + {\Omega_\alpha}^\beta + {\Sigma_\alpha}^\beta - {Y_\alpha}^\beta,
\end{equation}
where we have defined
\begin{multline}\label{ydef}
    {Y_\alpha}^\beta := k^\beta n_\rho \nabla_\alpha k^\rho + n^\beta k_\rho \nabla_\alpha k^\rho + k_\alpha n^\rho \nabla_\rho k^\beta + k_\alpha k^\beta n_\epsilon n^\rho \nabla_\rho k^\epsilon + k_\alpha n^\beta k_\rho n^\epsilon \nabla_\epsilon k^\rho + n_\alpha A^\beta + n_\alpha k^\beta n_\rho A^\rho  .
\end{multline}

Using the definition the expansion scalar \eqref{nullexp}, it can be seen that
\begin{equation}
    \frac{\Df \Theta}{\df \lambda} = k^\rho \nabla_\rho \left( {h^\alpha}_\beta \nabla_\alpha k^\beta \right) = \delta^\alpha_\beta k^\rho\nabla_\rho\nabla_\alpha k^\beta + k^\rho \nabla_\rho \left( n_\alpha A^\alpha - Q_{\alpha\beta\epsilon}k^\alpha k^\beta n^\epsilon \right).
\end{equation}
Using the Ricci identity \eqref{ricciid}, we can evaluate the first term above and write,
\begin{equation}\label{expeqnull1}
    \frac{\Df \Theta}{\df \lambda} = -R_{\alpha\beta}k^\alpha k^\beta + \nabla_\alpha A^\alpha - \nabla_\alpha k^\beta \nabla_\beta k^\alpha + k^\rho\nabla_\rho\left( n_\alpha A^\alpha\right) - k^\rho \nabla_\rho \left( Q_{\alpha\beta\epsilon}k^\alpha k^\beta n^\epsilon \right).
\end{equation}
In a similar manner as in the case of timelike curves, we can use equation \eqref{velgrad1} to show
\begin{equation}\label{sqterms}
    \nabla_\alpha k^\beta \nabla_\beta k^\alpha = \frac{1}{2}\Theta^2 - \Omega^2 + \Sigma^2 + {Y_\alpha}^\beta {Y_\beta}^\alpha,
\end{equation}
where we have defined $\Omega^2 := {\Omega_\alpha}^\beta {\Omega^\alpha}_\beta$, $\Sigma^2 := {\Sigma_\alpha}^\beta {\Sigma^\alpha}_\beta$. Using the definition in equation \eqref{ydef}, the last term on the right hand side turns out to be completely in terms of the path acceleration and the non-metricity
\begin{equation}\label{ysqterm}
    {Y_\alpha}^\beta {Y_\beta}^\alpha = -2n_\alpha A^\beta \nabla_\beta k^\alpha - \left( n_\alpha A^\alpha \right)^2 - Q_{\alpha\beta\rho}k^\beta k^\rho n^\epsilon \nabla_\epsilon k^\alpha - \left( Q_{\alpha\beta\rho}k^\alpha k^\beta n^\rho \right)^2. 
\end{equation}
Using equations \eqref{sqterms} and \eqref{ysqterm} in \eqref{expeqnull1}, we finally get
\begin{multline}\label{sachseq}
    \frac{\Df \Theta}{\df \lambda} = -R_{\alpha\beta}k^\alpha k^\beta - \frac{1}{2}\Theta^2 + \Omega^2 - \Sigma^2 + \nabla_\alpha A^\alpha + k^\rho\nabla_\rho\left( n_\alpha A^\alpha\right) + 2n_\alpha A^\beta \nabla_\beta k^\alpha + \left( n_\alpha A^\alpha \right)^2\\ - k^\rho \nabla_\rho \left( Q_{\alpha\beta\epsilon}k^\alpha k^\beta n^\epsilon \right) + Q_{\alpha\beta\rho}k^\beta k^\rho n^\epsilon \nabla_\epsilon k^\alpha + \left( Q_{\alpha\beta\rho}k^\alpha k^\beta n^\rho \right)^2 ,
\end{multline}
which is the Sachs optical equation within our geometry. Again, it is easy to see that in the absence of non-metricity, the above equation reduces to its usual form in the Riemannian case.

\subsubsection{Lagrangian Formulation}
Although both the Raychaudhuri and Sachs equations are first order differential equations, they can be converted into second order differential equations by relating the timelike and null expansion scalars to the fractional rate of change of cross sectional volume and area, respectively. Doing this enables one to define Lagrangians such that the corresponding Euler-Lagrange equations give the Raychaudhuri and Sachs equations. This has been previously done in Riemannian geometry \cite{chakra,horw,alsa}. A non-Riemannian extension for timelike curves was presented in \cite{Agashe_2024}. Working in Schr\"odinger geometry (fixed length vectors) enables us to extend this to null curves.

In order to relate the expansion scalar of null congruences to their cross sectional area, one first has to introduce the notion of cross section. A formal way of doing this can be found in \cite{poisson}. Following the same treatment, it can be shown that
\begin{equation}\label{areachange}
    \frac{1}{A} \frac{\Df A}{\df \lambda} = \Theta - \left( n_\alpha A^\alpha - Q_{\alpha\beta\rho} k^\alpha k^\beta n^\rho - {Q^\alpha}_{\alpha\rho}k^\rho \right),
\end{equation}
where $A$ is the cross sectional area of the congruence. Let us define a scalar function, $\phi$, such that
\begin{equation}
    \phi = {\rm e}^{\int \left(n_\alpha A^\alpha - Q_{\alpha\beta\rho} k^\alpha k^\beta n^\rho - {Q^\alpha}_{\alpha\rho}k^\rho \right) \df \lambda} \quad \Rightarrow \frac{\Df \ln \phi}{\df \lambda} = n_\alpha A^\alpha - Q_{\alpha\beta\rho} k^\alpha k^\beta n^\rho - {Q^\alpha}_{\alpha\rho}k^\rho.
\end{equation}
Using this, equation \eqref{areachange} becomes
\begin{equation}\label{areachange1}
    \frac{1}{(\phi A)}\frac{\Df (\phi A)}{\df \lambda} = \Theta.
\end{equation}
Taking a derivative with respect to $\lambda$ again, we get
\begin{equation}\label{areader}
    \frac{1}{(\phi A)}\frac{\Df^2 (\phi A)}{\df \lambda^2} - \left[\frac{1}{(\phi A)}\frac{\Df (\phi A)}{\df \lambda}\right]^2 = \frac{\Df \Theta}{\df \lambda}.
\end{equation}
Using equations \eqref{areachange1} and \eqref{areader} in the Sachs equation \eqref{sachseq}, we obtain
\begin{multline}\label{sachssecorder}
    \frac{1}{(\phi A)}\frac{\Df^2 (\phi A)}{\df \lambda^2} - \frac{1}{2}\left[\frac{1}{(\phi A)}\frac{\Df (\phi A)}{\df \lambda}\right]^2 = -R_{\alpha\beta}k^\alpha k^\beta + \Omega^2 - \Sigma^2 + \nabla_\alpha A^\alpha + k^\rho\nabla_\rho\left( n_\alpha A^\alpha\right) + 2n_\alpha A^\beta \nabla_\beta k^\alpha \\+ \left( n_\alpha A^\alpha \right)^2 - k^\rho \nabla_\rho \left( Q_{\alpha\beta\epsilon}k^\alpha k^\beta n^\epsilon \right) + Q_{\alpha\beta\rho}k^\beta k^\rho n^\epsilon \nabla_\epsilon k^\alpha + \left( Q_{\alpha\beta\rho}k^\alpha k^\beta n^\rho \right)^2.
\end{multline}
Thus, the Sachs equation is converted into a second order differential equation.  

Now, treating the scalar, $\phi A =: q$, as a dynamical degree of freedom, one can define a Lagrangian in the following manner
\begin{equation}
    \mathcal{L} = \frac{1}{2} \left[\frac{1}{q}\frac{\Df q}{\df \lambda}\right]^2 - V(q),
\end{equation}
where the first term on the left hand side can be interpreted as a kinetic term (square of first derivative) and $V(q)$ is a potential to be fixed later . It is easy to check that the Euler-Lagrange equation for the above Lagrangian is given by
\begin{equation}\label{eleq}
    \frac{1}{q}\frac{\Df^2 q}{\df \lambda^2} - \frac{1}{2} \left[ \frac{1}{q} \frac{\df q}{\df \lambda} \right]^2 = -\frac{\partial V}{\partial q}.
\end{equation}
Comparing equations \eqref{sachssecorder} and \eqref{eleq}, it is clear that the Euler-Lagrange equation above is precisely the Sachs equation with the potential defined using the following constraint equation
\begin{multline}
    -\frac{\partial V}{\partial q} = R_{\alpha\beta}k^\alpha k^\beta + \Omega^2 - \Sigma^2 + \nabla_\alpha A^\alpha + k^\rho\nabla_\rho\left( n_\alpha A^\alpha\right) + 2n_\alpha A^\beta \nabla_\beta k^\alpha + \left( n_\alpha A^\alpha \right)^2 \\- k^\rho \nabla_\rho \left( Q_{\alpha\beta\epsilon}k^\alpha k^\beta n^\epsilon \right) + Q_{\alpha\beta\rho}k^\beta k^\rho n^\epsilon \nabla_\epsilon k^\alpha + \left( Q_{\alpha\beta\rho}k^\alpha k^\beta n^\rho \right)^2.
\end{multline}

It is important to note here that although the Sachs equation can be written as the Euler-Lagrange equation for the above Lagrangian, it is not a dynamical equation in the sense that it does not tell us the dynamics of the space-time. The dynamics of the space-time are always given by a theory of gravity (field equations). However, in writing the Lagrangian formulation, we are treating the congruence itself as a dynamical system that is evolving according to the geometry it resides in. Therefore, the treatment above will prove useful in studies of geometric flows within the context of Schr\"odinger geometry. Moreover, starting with the Lagrangian given above, one can easily derive a Hamiltonian for such geometric flows which might prove to be useful in canonical quantization procedures.

\subsection{Cosmological Distances}
In cosmology, the light coming from distant sources holds the information about cosmological distances. This makes the distances closely related to the trajectories of bundles of photons. These photon bundles can be treated as null congruences. Then, the angular diameter distance, $D_A$, is proportional to the cross sectional area, $A$, of such congruences. In general, one can write
\begin{equation}
    A \propto D_A^2.
\end{equation}

In the usual formalism, since the cross sectional area is directly related to the expansion scalar $\Theta$, one can derive an equation for the angular diameter distances using the Sachs equation. However, in our case, we have seen in equation \eqref{areachange1} that the expansion scalar is actually related to the area weighted by a scalar factor, $\phi$, that is given in terms of the non-metricity $Q$ of the geometry. Therefore, we modify the above proportionality to write, 
\begin{equation}
    \phi A \propto \phi D_A^2.
\end{equation}
Let us define a `weighted' angular diameter distance as
\begin{equation}
    \Delta_A = \sqrt{\phi} D_A.
\end{equation}
Using this, the above proportionality becomes
\begin{align}
    \phi A &\propto \Delta_A ^2\\
    \Rightarrow\ \frac{\Df (\phi A)}{\df \lambda} &\propto 2 \Delta_A \frac{\Df \Delta_A}{\df \lambda}.
\end{align}
Combining the two proportionalities above, we get
\begin{equation}
    \frac{1}{(\phi A)}\frac{\Df (\phi A)}{\df \lambda} = \frac{2}{\Delta_A} \frac{\Df \Delta_A}{\df \lambda}.
\end{equation}
Using equation \eqref{areachange1} in this, we can write
\begin{equation}
    \frac{2}{\Delta_A} \frac{\Df \Delta_A}{\df \lambda} = \Theta.
\end{equation}
Again taking a derivative on both sides leads to
\begin{equation}
    \frac{2}{\Delta_A} \frac{\Df^2 \Delta_A}{\df \lambda^2} - \frac{1}{2} \left[ \frac{2}{\Delta_A} \frac{\Df \Delta_A}{\df \lambda} \right]^2 = \frac{\Df \Theta}{\df \lambda}.
\end{equation}
Using the two equations above in the Sachs equation \eqref{sachseq}, we finally get an equation relating the angular diameter distance to the geometric quantities
\begin{multline}
    \frac{\Df^2 \Delta_A}{\df \lambda^2} = \frac{1}{2} \left[ -R_{\alpha\beta}k^\alpha k^\beta + \Omega^2 - \Sigma^2 + \nabla_\alpha A^\alpha + k^\rho\nabla_\rho\left( n_\alpha A^\alpha\right) + 2n_\alpha A^\beta \nabla_\beta k^\alpha + \left( n_\alpha A^\alpha \right)^2 \right.\\ \left. - k^\rho \nabla_\rho \left( Q_{\alpha\beta\epsilon}k^\alpha k^\beta n^\epsilon \right) + Q_{\alpha\beta\rho}k^\beta k^\rho n^\epsilon \nabla_\epsilon k^\alpha + \left( Q_{\alpha\beta\rho}k^\alpha k^\beta n^\rho \right)^2 \right] \Delta_A.
\end{multline}
The above equation is quite general and holds for arbitrary cosmological space-times in Schr\"odinger geometry, but it could be specialized to the case of Yano-Schrödinger geometry by taking the non-metricity to be vectorial.

\section{Yano-Schrödinger Cosmology}\label{section4}

We turn our attention to the possibility of understanding the universe, from two distinct perspectives, using a Yano-Schrödinger connection. First, we fix the geometry and propose a model, in which dynamics of the non-metricity is imposed by a physically reasonable assumption. We then adopt a different perspective, inspired by metric-affine gravity, treating the connection and the metric to be related by the connection field equations. In this case, we take as matter a novel type of Perfect Hyperfluid, which sources the Yano-Schrödinger connection. In both approaches, we compare how well our proposed models fit the observational data of the Hubble function and the standard $\Lambda$CDM paradigm.

\subsection{Metric Approach}
Let us start the investigation by fixing the geometry, and consequently the Einstein field equations. Hence, we \textit{postulate} that
\begin{equation}
    R_{(\mu \nu)}-\frac{1}{2}R g_{\mu \nu}=8\pi T_{\mu \nu},
\end{equation}
where $R_{\mu \nu}$ and $R$ are the Ricci tensor and Ricci scalar of the Yano-Schrödinger geometry, respectively. Using equations \eqref{riccitensoryanoschrodinger} and \eqref{ricciscalaryanoschrodinger} we obtain
\begin{equation}\label{einsteinequation}
    \overset{\circ}{R}_{\mu \nu} - \frac{1}{2} g_{\mu \nu} \overset{\circ}{R} -\frac{5}{4} g_{\mu \nu} \overset{\circ}{\nabla}_{\alpha} \pi^{\alpha}+\frac{1}{4} \left(\overset{\circ}{\nabla}_{\mu} \pi_{\nu} + \overset{\circ}{\nabla}_{\nu} \pi_{\mu} \right)+ \frac{5}{8} g_{\mu \nu} \pi^{\alpha} \pi_{\alpha}-\frac{1}{4} \pi_{\nu} \pi_{\mu}=8 \pi T_{\mu \nu}.
\end{equation}

In the limit $\pi \to 0$, i.e. by passing to semi-Riemannian geometry, we recover Einstein's original equations. The extra five terms in our equations come from the considered type of non-metricity, and we think of these as a type of dark energy, which is introduced by the modified geometry. Right now, it might not be obvious why we think of these terms this way. However, this idea will make more sense when we apply the field equations to a specific model of the universe.
\subsubsection{The Generalized Friedmann Equations}\label{CosmologyFriedmann}
The starting point to derive the Friedmann equations in Yano-Schrödinger geometry is the generalized Einstein equation \eqref{einsteinequation}.
We consider an isotropic, homogeneous and spatially flat FLRW metric, which is described by
\begin{equation}
    ds^2=-dt^2+a^2(t) \delta_{ij} dx^i dx^j,
\end{equation}
where latin indices take the values $1,2,3$. For the matter content, we choose a perfect fluid 
\begin{equation}
    T_{\mu \nu}=\rho u_\mu u_\nu +p(u_\mu u_\nu+g_{\mu \nu}).
\end{equation}
The problem is taken into account in a comoving coordinate system, where 
\begin{equation}
    u_{\nu}=(1,0,0,0) \iff u^{\nu}=(-1,0,0,0).
\end{equation}
As we are in a highly symmetric case, we choose
\begin{equation}
    \pi_\nu=(\Pi(t),0,0,0) \iff \pi^{\nu}=(-\Pi(t),0,0,0).
\end{equation}
From now on, we won't write out explicitly the time-dependence of $a(t)$. With the above  assumptions, a derivation detailed in  \ref{appendixB} gives the following evolution equations
\begin{equation}\label{Friedmann1}
    3H^2=8\pi \rho +\frac{3}{4} \dot \Pi +\frac{15}{4} H \Pi - \frac{3}{8} \Pi^2=8\pi(\rho+\rho_{DE}),
\end{equation}
\begin{equation}\label{Friedmann2}
        3H^2+2 \dot{H}=-8\pi p +\frac{5}{4} \dot \Pi +\frac{13}{4} H \Pi - \frac{5}{8} \Pi^2=-8\pi(p+p_{DE}),
\end{equation}
where we have defined the additional terms coming from torsion 
\begin{equation}
    \rho_{DE}=\frac{1}{32\pi} \left(3 \dot \Pi + 15 H \Pi - \frac{3}{2} \Pi^2 \right), \; \; p_{DE}=\frac{1}{32\pi} \left(-5 \dot \Pi -13 H \Pi +\frac{5}{2} \Pi^2\right).
\end{equation}
The energy conservation equation
\begin{equation}\label{energyequation}
    \dot{\rho}+\dot{\rho}_{DE}+3H \left( \rho +\rho_{DE} + p + p_{DE} \right)=0
\end{equation}
can be equivalently reformulated as
\begin{equation}
    \dot\rho +3H(\rho +p) + \frac{1}{32\pi}\frac{d}{dt}  \left(3 \dot \Pi + 15 H \Pi - \frac{3}{2} \Pi^2 \right)+\frac{3H}{32 \pi} \left( -2 \dot \Pi +2 H \Pi +\Pi^2 \right)=0.\end{equation}
As an indicator of the accelerating/decelerating nature of the cosmological expansion we use the deceleration parameter
\begin{equation}
    q=\frac{d}{dt}\frac{1}{H}-1=-\frac{\dot H}{H^2}-1.
\end{equation}
In our concrete case, using the Friedmann equations, this can be computed:
\begin{equation}
    q=\frac{1}{2}+\frac{3}{2} \frac{p + p_{DE}}{\rho+\rho_{DE}}=\frac{1}{2}+\frac{3}{2}\frac{8 \pi p -\frac{1}{4} \left( 5 \dot{\Pi} + 13 H \Pi -\frac{5}{2} \Pi^2\right)}{8 \pi \rho + \frac{3}{4} \left( \dot \Pi + 5 H \Pi - \frac{1}{2} \Pi^2 \right)}.
\end{equation}
\subsubsection{De Sitter Solutions}
De Sitter solutions are of particular interest in cosmology, as they are attractors of cosmological models, described by the constancy of the Hubble function. For our theory, considering dust matter with $p=0, \rho \neq 0$ and $H=H_0=\text{constant}$, the Friedmann equations become
\begin{equation}
    3 H_0^2=8 \pi \rho +\frac{3}{4} \dot \Pi +\frac{15}{4} H_0 \Pi - \frac{3}{8} \Pi^2,
\end{equation}
\begin{equation}\label{friedmannsecond}
    3H_0^2=\frac{5}{4} \dot \Pi +\frac{13}{4} H_0 \Pi - \frac{5}{8} \Pi^2.
\end{equation}
The second Friedmann equation \eqref{friedmannsecond} admits an analytical solution, given by
\begin{equation}
    \Pi(t)=\frac{2 H_0 \left(-2+3 e^{\frac{7}{5}H_0(t+10t_0)} \right)}{-1+5 e^{\frac{7}{5}H_0(t+10t_0)}}.
\end{equation}
In the large time limit, we obtain
\begin{equation}
    \lim_{t \to \infty} \Pi(t)=\frac{6}{5} H_0.
\end{equation}
Then, from the first Friedmann equation, we obtain the time evolution of $\rho$, given as
\begin{equation}
    \rho(t)=\frac{3 H_0^2}{8 \pi} \left(1 - \frac{49  e^{\frac{7}{10}H_0(t+10t_0)}}{10\left( 1-5 e^{\frac{7}{5}H_0(t+10 t_0)} \right)} - \frac{10 \left(-2+3 e^{\frac{7}{5}H_0(t+10t_0)} \right)}{-4+20 e^{\frac{7}{5}H_0(t+10t_0)}} +  \frac{ -2+3 e^{\frac{7}{5}H_0(t+10t_0)} }{2\left(-1+5 e^{\frac{7}{5} H_0 (t+10t_0)} \right)^2}\right).
\end{equation}
In the late universe, the above solution converges to
\begin{equation}
    \lim_{t \to \infty} \rho(t)=\frac{3 H_0^2}{8 \pi} \left( 1 + \frac{49}{50} - \frac{3}{2} \right)=\frac{9}{50 \pi} H_0^2,
\end{equation}
which indicates that the energy condition $\rho>0$ is not violated at large time intervals.

If $\rho=0$, but $p \neq 0$, the first Friedmann equation \eqref{Friedmann1} takes the form
\begin{equation}\label{admitsolution}
    3H_0^2=\frac{3}{4} \dot \Pi +\frac{15}{4} H_0 \Pi -\frac{3}{8} \Pi^2. 
\end{equation}
Equation \eqref{admitsolution} admits an analytical solution, which reads
\begin{equation}
    \Pi(t)=\sqrt{17} H_0 \tanh \left(\frac{1}{2} \left(\sqrt{17} H_0 t_0 - \sqrt{17} H_0 t\right) \right)+5H_{0},
\end{equation}
where $t_0$ is an arbitrary integration constant. In the large time limit, we have
\begin{equation}
    \lim_{t \to \infty} \Pi(t)=-\sqrt{17}H_0+5 H_{0},
\end{equation}
that is,  $\Pi(t)$ takes positive values when $t$ is large. Then, from equation \eqref{Friedmann2}, we obtain the time evolution of the pressure
\begin{equation}
    p(t)=\frac{H_0^2}{8\pi} \Bigg \{-3  -\frac{85\sqrt{17}}{16}  H_0 \operatorname{sech}^2 \left( t_0 - t \right)  + \left( \frac{17}{2} H_0 \tanh \left(  t_0 - t\right) +5\right)\left( \frac{51}{8} -\frac{85}{16} H_0 \tanh \left( t_0 - t\right)  \right)\Bigg \}.
\end{equation}
In the large time limit, the pressure becomes a negative constant fully determined by the Hubble parameter
\begin{equation}
    \lim_{t \to \infty}p(t)=\frac{1}{8 \pi} \left( -3 H_0^2 -\frac{13}{4} \sqrt{17}H_0^2 +\frac{65}{4} H_0^2 -\frac{5}{8} \left(-\sqrt{17} H_0 + 5H_0 \right)^2\right)=\frac{-13+3 \sqrt{17}}{8\pi} H_0^2.
\end{equation}
\subsubsection{Dimensionless and Redshift Representations}
Motivated by comparing our theory with observational data for the Hubble parameter, we first rewrite the Friedmann equations in a dimensionless form by introducing a quintuple $(h,\tau,\Psi,r,P)$ defined as
\begin{equation}
    H=H_0h,\tau=H_0t,\Pi=H_0 \Psi,\rho=\frac{3H_0^2}{8 \pi}r,p=\frac{3 H_0^2}{8 \pi} P.
\end{equation}
Using this parametrization, the dimensionless Friedmann equations are given by
\begin{equation}\label{redshift1}
    h^2=r+\frac{1}{4} \frac{d \Psi}{d \tau}+\frac{5}{4} h \Psi - \frac{1}{8} \Psi^2,
\end{equation}
\begin{equation}\label{redshift2}
    3h^2+2 \frac{d h}{d \tau}=-3P+\frac{5}{4}\frac{d \Psi}{d \tau} +\frac{13}{4}  h\Psi - \frac{5}{8} \Psi^2.
\end{equation}
The energy balance equation \eqref{energyequation} takes the dimensionless form
\begin{equation}\label{redshift3}
\frac{3}{8 \pi}  \frac{dr}{d \tau} +3h\left(\frac{3}{8 \pi} r +\frac{3}{8\pi}P \right) + \frac{1}{32\pi}\frac{d}{d\tau}  \left(3 \frac{d \Psi}{d \tau} + 15 h \Psi - \frac{3}{2} \Psi^2 \right) + \frac{3h}{32 \pi} \left( -2 \frac{d \Psi}{d \tau} +2 h \Psi +\Psi^2\right)=0.
\end{equation}
A straightforward algebraic simplification yields
\begin{equation}
    \frac{dr}{d\tau} +3h(r+P)+\frac{1}{12} \frac{d}{d \tau} \left(3 \frac{d \Psi}{d \tau} + 15 h \Psi -\frac{3}{2} \Psi^2 \right)+\frac{h}{4} \left(-2 \frac{d \Psi}{d \tau}  + 2h \Psi +\Psi^2 \right)=0.
\end{equation}
Having the dimensionless form, we move to redshift space, as the observational data is given by the redshift variable $z$, implicitly defined through
\begin{equation}
    1+z=\frac{1}{a}, \; \; \text{thus} \; \; \frac{d}{d \tau}=-(1+z)h(z)\frac{d}{dz}.
\end{equation}
In redshift space, equations \eqref{redshift1}, \eqref{redshift2}, and \eqref{redshift3} take the following form
\begin{equation}\label{substitutehere}
    h^2(z)=r(z)-\frac{1}{4}(1+z)h(z)\frac{d \Psi(z)}{dz}+\frac{5}{4} h(z) \Psi(z) - \frac{1}{8} \Psi^2(z)
\end{equation}
\begin{equation}
    3h^2(z)-2(1+z)h(z) \frac{dh(z)}{dz}=-3P(z)-\frac{5}{4}(1+z)h(z)\frac{d \Psi(z)}{dz}+\frac{13}{4} h(z) \Psi(z) - \frac{5}{8} \Psi^2(z)
\end{equation}
\begin{equation}
\begin{aligned}
   & -(1+z)h(z) \frac{dr(z)}{dz}+3h(z)(r(z)+P(z))\\
    &-\frac{1}{12}(1+z)h(z) \frac{d}{dz}\left(-3(1+z)h(z) \frac{d\Psi(z)}{dz} +15h(z)\Psi(z)-\frac{3}{2}\Psi^2(z) \right)\\
    &+\frac{h(z)}{4} \left(2(1+z)h(z) \frac{d\Psi(z)}{dz} +2h(z)\Psi(z)+ \Psi^2(z) \right)=0
\end{aligned}
\end{equation}
In the upcoming section, we will test Yano-Schrödinger gravity, by a detailed comparison with the standard $\Lambda$CDM model, which we recall in the following, and a small sample of observational data, obtained for the Hubble parameter. 

In the $\Lambda $CDM model the Hubble function is expressed as
\begin{equation}
H=H_{0}\sqrt{\frac{\Omega _{m}}{a^{3}}+\Omega _{\Lambda }}=H_{0}\sqrt{\Omega
_{m}(1+z)^{3}+\Omega _{\Lambda }},
\end{equation}
where $\Omega _{m}=\Omega _{b}+\Omega _{DM}$,  $\Omega _{b}=\rho
_{b}/\rho _{cr}$, $\Omega _{DM}=\rho _{DM}/\rho _{cr} $ and $\Omega
_{\Lambda }=\Lambda /\rho _{cr}$, where $\rho_{cr}$ is the critical density
of the Universe. $\Omega _{b}$, $\Omega _{DM}$ and $\Omega _{DE}$ represent
the density parameters of the baryonic matter, dark matter, and dark energy,
respectively. The deceleration parameter is given by
\begin{equation}
q(z)=\frac{3(1+z)^{3}\Omega _{m}}{2\left[ \Omega _{\Lambda }+(1+z)^{3}\Omega
_{m}\right] }-1.
\end{equation}

For our analysis of the matter and dark energy density parameters in the $\Lambda$CDM model, we will use the following values: $\Omega_{\Lambda}=0.6847$, $\Omega_{m}=0.3166$, based on \cite{1g}. The observational data, together with the error bars for the Hubble parameter is obtained from \cite{Bou}.

For an in depth comparison with $\Lambda$CDM model, we'll also use the $Om(z)$ diagnostic \cite{Sahni}, a crucial tool for differentiating alternative cosmological models. The $Om(z)$ function is defined as
\begin{equation}
Om (z)=\frac{H^2(z)/H_0^2-1}{(1+z)^3-1}=\frac{h^2(z)-1}{(1+z)^3-1}.
\end{equation}
In the case of the $\Lambda$CDM model, $Om(z)$ is a constant equal to the present day matter density $r(0)=0.3166$. However, in other theories of gravity that differ from the $\Lambda$CDM model, changes in the value of $Om(z)$ over time indicate different types of cosmic evolution. Specifically, if $Om(z)$ increases (positive slope), it suggests a phantom-like evolution. On the other hand, if $Om(z)$ decreases (negative slope), it points to  quintessence-like dynamics.

\subsubsection{Cosmological Model}
In the following analysis, we will consider matter to be a pressureless dust with $p=0$. So far, we lack a dynamical equation for $\Psi$, which leaves our system underdetermined. To be able to solve the evolution equations, we have to provide an equation of state for dark energy. We propose the physically reasonable assumption that matter is conserved
\begin{equation}\label{thisone1}
    \dot \rho+3H \rho=0,
\end{equation}
which immediately implies by the energy conservation equation that
\begin{equation}\label{thisone2}
   \frac{d}{dt}\left(\dot \Pi+5 H \Pi - \frac{1}{2} \Pi^2 \right) -2 H\dot \Pi +2H^2 \Pi+H \Pi^2=0.
\end{equation}
In terms of redshift variables, equations \eqref{thisone1} and \eqref{thisone2} can be equivalently written as
\begin{equation}\label{integratethis}
    -(1+z)h(z) \frac{d r(z)}{dz} +3 h(z) r(z)=0,
\end{equation}
\begin{equation}
\begin{aligned}
    -(1+z)h(z)\frac{d}{dz} \left(-(1+z)h(z) \frac{d \Psi(z)}{dz} + 5h(z) \Psi(z) -\frac{1}{2} \Psi^2(z) \right)\\
    +2h(z)(1+z)h(z) \frac{ d \Psi(z)}{dz} +2h^2(z) \Psi(z)+h(z) \Psi^2(z)=0.
\end{aligned}
\end{equation}
As we treated matter being a pressureless dust, equation \eqref{integratethis} can be integrated to obtain
\begin{equation}
    r(z)=r(0)(1+z)^3,
\end{equation}
where $r(0)$ is the present day matter energy density. Hence, equation \eqref{substitutehere} yields
\begin{equation}
    h^2(z)=r(0)(1+z)^3-\frac{1}{4}(1+z) h(z) \frac{d \Psi(z)}{dz}+\frac{5}{4} h(z) \Psi(z) -\frac{1}{8} \Psi^2(z).
\end{equation}

Therefore, the evolution equations of a Yano-Schrödinger cosmological model with pressureless dust, where both matter and dark matter are conserved take the final form
\begin{equation}\label{cosmology1}
    3h^2(z)-2(1+z)\frac{dh(z)}{dz}=-\frac{5}{4}(1+z)h(z) \frac{d \Psi(z)}{dz}+\frac{13}{4}h(z) \Psi(z)-\frac{5}{8} \Psi^2(z),
\end{equation}
\begin{equation}\label{cosmology2}
    -(1+z)h(z) \frac{d}{dz} \left( 4h(z)^2-4r(0)(1+z)^3 \right)-8 h^3(z)+8r(0)(1+z)^3 h(z) +12 h^2(z) \Psi(z)=0.
\end{equation}
The system of differential equations \eqref{cosmology1}-\eqref{cosmology2} has to be integrated with the initial conditions $h(0)=1$ and $\Psi(0):=\Psi_0$.

\begin{figure*}[htbp]
\centering
\includegraphics[width=0.490\linewidth]{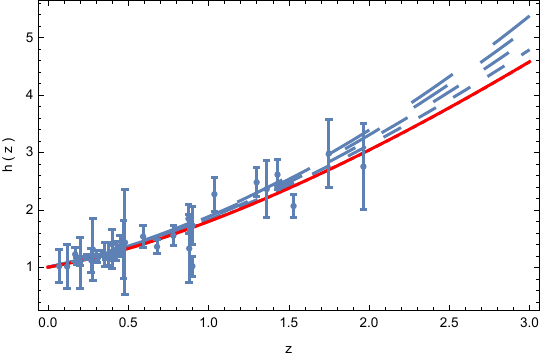} %
\includegraphics[width=0.490\linewidth]{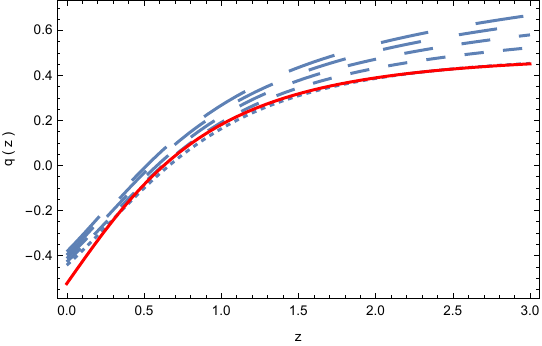}
\caption{Variations of the dimensionless Hubble function $h(z) $(left panel), and of the deceleration parameter $q(z)$ (right panel) for Model 1 with initial conditions $\Psi(0)=0.51$ (dotted curve), $\Psi(0)=0.52$ (short dashed curve), $\Psi(0)=0.53$ (dashed curve) , $\Psi(0)=0.54$ (long dashed curve), $\Psi(0)=0.55$ (ultra-long dashed curve), respectively. The observational data for the Hubble function are represented with their error bars, while the red curve depicts the predictions of the $\Lambda$CDM model.}
\label{fig3}
\end{figure*}
\begin{figure*}[htbp]
\centering
\includegraphics[width=0.490\linewidth]{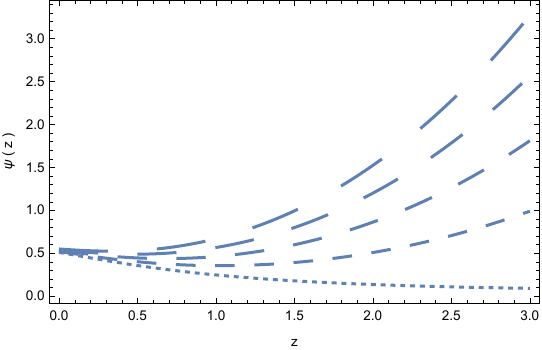} %
\includegraphics[width=0.490\linewidth]{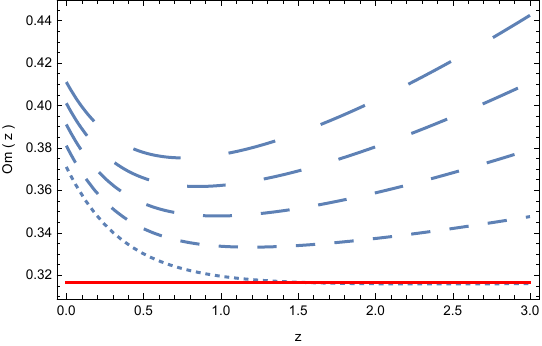}
\caption{Variations of the dimensionless non-metricity $\Psi(z)$ (left panel), and of the $Om(z)$ function (right panel) for Model 2 with initial conditions $\Psi(0)=0.51$ (dotted curve), $\Psi(0)=0.52$ (short dashed curve), $\Psi(0)=0.53$ (dashed curve) , $\Psi(0)=0.54$ (long dashed curve), $\Psi(0)=0.55$ (ultra-long dashed curve), respectively. The red curve represents the predictions of the $\Lambda$CDM model.}
\label{fig4}
\end{figure*}

As can be seen on Figure \ref{fig3}, for the considered range of parameters, the dark energy model with conserved matter density can reproduce the observational data for the Hubble parameter. Moreover, for redshifts $0\leq z \leq 2$, the predictions of the proposed model are in almost exact accordance with those of the standard $\Lambda$CDM paradigm. For $z>2$, there is a small deviation in both the values of the Hubble function $h(z)$ and the deceleration parameter $q(z)$, our model predicting slightly higher values. Up to redshifts $z \simeq 1,5$, the two models basically coincide.

Figure \ref{fig4} shows that the non-metricity vector $\Psi(z)$ for the initial conditions $\Psi(0)>0.51$ is monotone increasing after $z>0.5$ and takes positive values during the cosmological evolution, but its behaviour depends on the initial condition chosen. For instance, in the case of the initial value  $\Psi(0)=0.51$, $\Psi(z)$ is still positive, but monotone decreasing. From the same figure, it can be seen that the $Om(z)$ function is very different from the diagnostic function of the $\Lambda$CDM model, in which $Om(z)$ is a constant. In our case, at lower redshifts, up to $z\simeq 1$, it has a negative slope,  while at higher redshifts the slope becomes positive, indicating a transition from quintessence-like dynamics to a phantom-like evolution. Nevertheless, for the initial condition $\Psi(0)=0.51$, such a transition does not appear, and the model's $Om(z)$ function converges to that of $\Lambda$CDM.


\subsection{Palatini Approach}
In the previous subsection, we postulated the connection, obtained the corresponding field equations, and explored the cosmological implications of the proposed theory. Here we take a metric-affine approach, by first finding the hypermomentum tensor, which sources Schrödinger connections. Having found the hypermomentum tensor, we propose a Lagrangian formulation of a perfect Schrödinger hyperfluid, following the recent work \cite{Iosifidis_2020fluids}.
    
The action considered in this subsection is the standard Einstein $f(R)=R$ action, with the only caveat being that we do not fix $R$ to be the Ricci scalar of the Yano-Schrödinger geometry, but this will be a consequence of the connection field equations. Moreover, we  allow matter to couple to the connection. Of course this part is quite essential and introduces the desired hypermomentum that sources this type of non-metricity.
In this case, our action reads
\begin{equation}\label{proposedaction}
    S=\frac{1}{2 \kappa} \int \sqrt{-g} R+ S_{M}(g,\Gamma,\Phi),
\end{equation}
where the first term is the usual Einstein-Hilbert action and $S_M$ denotes the matter sector. We note that in the matter sector, we allow a matter-geometry coupling, which could yield a non-trivial hypermomentum tensor. This is distinct to the approach being considered in \cite{HarkoSchr}, where hypermomentum was not taken into account.

Variation of the action  \eqref{proposedaction} with respect to the metric and the connection give the following field equations
\begin{equation} \label{connectionfield}
    R_{(\mu \nu)} - \frac{1}{2}R g_{\mu \nu}=\kappa T_{\mu \nu},
 \; \; 
    \tensor{P}{_\lambda ^{(\mu \nu)}}=\kappa \tensor{\Delta}{_\lambda^{(\mu \nu)}},
\end{equation}
where the energy momentum tensor $T$, and hypermomentum tensor $\Delta$ are given by
\begin{equation}\label{hyp}
        T^{\mu \nu}=\frac{2}{\sqrt{-g}} \frac{\delta \left(\sqrt{-g} \mathcal{L}_{M} \right)}{\delta g_{\mu \nu}},\; \; 
    \tensor{\Delta}{_\lambda^{\mu \nu}}=-\frac{2}{\sqrt{-g}} \frac{ \delta \left(\sqrt{-g} \mathcal{L}_{M} \right)}{\delta \tensor{\Gamma}{^\lambda _\mu _\nu}},
\end{equation}
and the Palatini tensor is defined as
\begin{equation}
     \tensor{P}{_\lambda^{(\mu \nu)}}=\frac{1}{2} Q_{\lambda} g^{\mu \nu}- \tensor{Q}{_\lambda ^{\mu \nu}}+ \left(q^{(\mu}-\frac{1}{2} Q^{(\mu} \right) \tensor{\delta}{_\lambda ^{\nu )}}, \; \; \text{with} \; \; Q_{\lambda}:=Q_{\lambda \mu \nu} g^{\mu \nu}\  {\text{and}} \; \; q_{\nu}:=Q_{\lambda \mu \nu}g^{\lambda \mu}.
\end{equation}
Note that the connection field equations \eqref{connectionfield} relate the non-metricity with the hypermomentum of matter. Namely, we can determine the type of hypermomentum, which sources a Schrödinger-type non-metricity.

\subsubsection{Yano-Schrödinger Hyperfluid Cosmology}

Let us now derive the form of the hypermomentum tensor that sources such a Schrödinger-Yano (i.e. length preserving) non-metricity. Firstly, let us briefly introduce the concept of hyperfluid. In Metric-Affine Gravity apart from the usual energy-momentum tensor one also encounters the hypermomentum tensor which is formally defined as the variation of the matter action with respect to the affine connection  (see eq. (\ref{hyp})). In general the hypermomentum describes the microproperties of matter and connects in a nice way the generalized geometry to the microstructure \cite{hehl1976hypermomentum}. For Friedmann-like Cosmological settings the hypermomentum has to be homogeneous and isotropic. It turns out that, in such a highly symmetric space, the latter is described by 5 functions of time and has the covariant form \cite{Iosifidis_2020fluids}
\begin{equation}
    \Delta_{\lambda \mu \nu}=\widetilde{\omega} u_\lambda u_\mu u_\nu + \psi u_\lambda g_{\mu \nu}+\phi u_\nu g_{\lambda \mu}+\chi u_\mu g_{\lambda \nu} +\zeta \tilde{\epsilon}_{\lambda \mu \nu \rho } u^{\rho },
\end{equation}
where $\widetilde{\omega},\psi,\phi$ and $\chi$ are functions of $t$. It  can equivalently be rewritten in a $3+1$ decomposition as
\begin{equation}\label{eq:HMFLRW}
    \Delta_{\lambda \mu \nu}=\omega u_\lambda u_\mu u_\nu + \psi u_\lambda h_{\mu \nu}+\phi u_\nu h_{\lambda \mu}+\chi u_\mu h_{\lambda \nu} +\zeta \tilde{\epsilon}_{\lambda \mu \nu \rho } u^{\rho } \; \; \text{with} \; \; \omega=\widetilde{\omega}-\phi-\psi-\chi.
\end{equation}
The canonical and metrical energy-momentum tensors take the perfect fluid form
\begin{equation}\label{eq:PfluidFLRW}
    t_{\mu \nu}=\rho_c u_{\mu} u_{\nu}+ p_c h_{\mu \nu}, \; \;T_{\mu\nu}=\rho u_\mu u_\nu + p h_{\mu\nu}.
\end{equation}
The above 2 energy related tensors along with the conservation laws of Metric-Affine Gravity, describe the behaviour of the fluid which is dubbed Perfect Cosmological Hyperfluid. 

 We now ask the question: What type of a hyperfluid sources a length-preserving non-metricity? Given the isotropic form of non-metricity, the above isotropic form of hypermomentum and the connection field equations (\ref{connectionfield}), a simple algebraic calculation detailed in \ref{appendixC} yields
\beq \label{connectionfieldequations}
\Delta_{\lambda \mu\nu}=\frac{D(t)}{2 \kappa}\left[ (n-3)u_{\lambda}h_{\mu\nu}+(4-2n) h_{\lambda(\mu}u_{\nu)}+(n-1)u_{\lambda}u_{\mu}u_{\nu} \right].
\eeq
We can thus easily identify the hyperfluid functions, which source the desired non-metricity
\beq
\phi=\chi=\frac{2-n}{2 \kappa}D \;, \; \;\psi=\frac{(n-3)}{2 \kappa }D\;, \;\;\; \omega=\frac{(n-1)}{2 \kappa }D\;, \;\;\; \zeta=0.
\eeq
This allows us to give a proper mathematical definition to a Yano-Schrödinger hyperfluid.
\begin{definition}
  A Perfect Hyperfluid $\left(T_{\mu \nu}, t_{\mu \nu},\Delta_{\lambda \mu \nu}  \right)$ is called a \textbf{Yano-Schrödinger hyperfluid} if its hypermomentum tensor takes the form
\beq 
\Delta_{\lambda \mu\nu}=\frac{D(t)}{2 \kappa}\left[ (n-3)u_{\lambda}h_{\mu\nu}+(4-2n) h_{\lambda(\mu}u_{\nu)}+(n-1)u_{\lambda}u_{\mu}u_{\nu} \right],
\eeq
where $D(t)$ is the smooth function describing a length-preserving non-metricity, i.e. a Yano-Schrödinger connection.
\end{definition}
\begin{remark}
    A Yano-Schrödinger hyperfluid sources a Yano-Schrödinger connection iff the gravitational sector is given by the Ricci scalar only, i.e. $f(R)=R$.
\end{remark}
We are now ready to provide a complete variational description of the Yano-Schrödinger hyperfluid. Following the ideas developed in \cite{iosifidis2023hyperhydrodynamics,Brown_1993} , we  propose 
\begin{equation}
I= \int d^n x \Big[{J}^\mu\Big( \varphi_{,\mu} + s\theta_{,\mu} + 
\beta_K \alpha^K{}_{,\mu}\Big) 
-\frac{\sqrt{-g}}{2} \left(2\rho(n,s,D) +
\frac{(-3 n +7)}{2 }Q_{\mu}D^{\mu}+(n-3) q_\mu D^\mu \right) \Big ]\label{actionhyp}, 
\end{equation}
where $D^{\mu}=\frac{D}{2 \kappa} u^{\mu}$ and $K=1,2,3$.This action describes a Yano-Schrödinger hyperfluid, as detailed in appendix \ref{appendixD}.

 Varying \eqref{actionhyp} with respect to the metric gives the metrical energy-momentum tensor
 \begin{equation}
     T_{\mu \nu}=\rho u_{\mu}u_{\nu}+p h_{\mu\nu}+g_{\mu\nu}\frac{(3 n -7)}{2}\frac{1}{\sqrt{-g}}\partial_{\alpha}(\sqrt{-g}D^{\alpha})+(n-3) \Big[ \xi_{(\mu}D_{\nu)}-\nabla_{(\mu}D_{\nu)} \Big],
 \end{equation}
where 
\beq
p:=n \frac{\partial \rho}{\partial n}+D \frac{\partial \rho}{\partial D}-\rho
\eeq
 and we have also abbreviated $\xi_{\mu}=-q_{\mu}+\frac{1}{2}Q_{\mu}$. Using the Cosmological ansatz the above energy-momentum tensor takes the form
 \beq
  T_{\mu \nu}=\rho_{eff} u_{\mu}u_{\nu}+p_{eff} h_{\mu\nu}
 \eeq
with
\begin{equation}
\rho_{eff} = \rho - \frac{3n - 7}{2} \left( \partial_t \widetilde{D} + (n-1)H \widetilde{D} \right) + (n-3) \left( \partial_t \widetilde{D} - (n-1)\widetilde{D}^{2} \right)
\end{equation}
\begin{equation}
p_{eff}=p+\frac{3n-7}{2}\Big(\partial_t{\widetilde{D}}+(n-1)H\widetilde{D} \Big)-(n-3)\widetilde{D}\left( H-\widetilde{D}\right),
\end{equation}
where we have defined 
\begin{equation}
    \widetilde{D}(t):=-\frac{D(t)}{2 \kappa}.
\end{equation}Using the notation of  Section \eqref{CosmologyFriedmann}, the Friedmann equations of a Yano-Schrödinger hyperfluid read
\begin{equation}
    3H^2=\kappa \rho+\frac{3}{2} \dot \Pi + \frac{15}{2} H \Pi -\frac{9}{8} \Pi^2,
\end{equation}
\begin{equation}
    2 \dot H + 3H^2=-\kappa p + \frac{5}{2} \dot \Pi +4 H \Pi -\frac{3}{8} \Pi^2.
\end{equation}
     Introducing the hypermomentum contributions
\begin{equation}
    \rho_{hyp}=\frac{1}{\kappa} \left( \frac{3}{2} \dot \Pi + \frac{15}{2} H \Pi  - \frac{9}{8} \Pi^2 \right), \; \; p_{hyp}=\frac{1}{\kappa} \left( - \frac{5}{2} \dot \Pi -4 H \Pi + \frac{3}{8} \Pi^2 \right),
\end{equation}
we can rewrite them as
\begin{equation}
    3H^2=\kappa (\rho + \rho_{hyp}),
\end{equation}
\begin{equation}
    2\dot{H}+3H^2=-\kappa(p+p_{hyp}).
\end{equation}
The continuity equation takes the modified form
\begin{equation}
    \dot \rho + \dot \rho_{hyp}+3H(\rho+\rho_{hyp}+p+p_{hyp})=0.
\end{equation}
In the following, we assume that the matter density is connection independent, that is 
\begin{equation}
    \frac{\partial \rho}{\partial D}=0.
\end{equation}
To close the system, we have to impose equations of state relating the ordinary matter density to pressure and the hypermomentum contributions. We assume a pressureless hyperfluid, where the ordinary matter density is conserved:
\begin{equation}
    p=0, \; \; \dot \rho +3H(\rho+p)=0.
\end{equation}
The dynamics for $\Pi$ is therefore encoded in
\begin{equation}
    \frac{3}{2} \ddot \Pi + \frac{15}{2} \dot H \Pi + \frac{15}{2} H \dot \Pi -\frac{9}{4} \Pi \dot \Pi + 3H \left (-\dot \Pi +\frac{7}{2} H \Pi -\frac{3}{4} \Pi^2 \right)=0.
\end{equation}
To directly compare with the observational data, we rewrite everything in redshift variables in two steps as before. In the first step, we introduce
\begin{equation}
H=H_0h, \tau=H_0 t, \Pi=H_0 \Psi, \rho=\frac{3H_0^2}{\kappa} r, p=\frac{3H_0^2}{\kappa} P,u=\frac{d \Psi}{d \tau}.
\end{equation}
In these variables, the Friedmann equations of a barotropic Yano-Schrödinger hyperfluid are given by
\begin{equation}\label{Friedmannhyper}
    3h^2=3r+\frac{3}{2} u + \frac{15}{4} h \psi - \frac{9}{8} \Psi^2,
\end{equation}
\begin{equation}
    2\frac{dh}{d\tau}+3h^2=  \frac{5}{2}u+4h \Psi - \frac{3}{8} \Psi^2.
\end{equation}
They are supplied with the conservation equation
\begin{equation}\label{hyperfluid1}
\frac{3}{2} \frac{ du}{d \tau}+\frac{15}{2} \frac{d h}{d \tau} \Psi + \frac{15}{2} h u - \frac{9}{4} \psi u - 3h u + \frac{21}{2} h^2 \Psi -\frac{9}{4} h \Psi^2=0.
\end{equation}
In the redshift representation, the evolution equations take the form
\begin{equation}\label{hyperfluid2}
    (1+z)h(z) \frac{d \Psi}{dz}+u(z)=0,
\end{equation}
\begin{equation}
    -2(1+z)h(z) \frac{dh(z)}{dz}+3h(z)^2 - \frac{5}{2} u(z) -4h(z) \Psi(z) +\frac{3}{8} \Psi^2(z)=0,
\end{equation}
\begin{equation}\label{hyperfluid3}
\begin{aligned}
    -\frac{3}{2}(1+z)h(z) \frac{du(z)}{dz}&-\frac{15}{2}(1+z)h(z) \frac{d h(z)}{dz} \Psi(z)+\frac{15}{2} h(z) u(z)-\frac{9}{4}\Psi(z) u(z) \\
    &-3h(z)u(z)+\frac{21}{2} h^2(z) \Psi(z) - \frac{9}{4} h(z) \Psi^2(z)=0.
\end{aligned}
\end{equation}
The system of equations \eqref{hyperfluid1}-\eqref{hyperfluid3} has to be solved with the initial conditions $h(0)=1, \Psi(0)=\Psi_0, u(0)=u_0$. Having solved the system, the matter density is obtained from the first Friedmann equation  \eqref{Friedmannhyper} as
\begin{equation}
    r(z)=h^2(z)-\frac{1}{2} u(z) - \frac{5}{4} h(z) \Psi(z) + \frac{3}{8} \Psi^2(z).
\end{equation}

\begin{figure*}[htbp]
\centering
\includegraphics[width=0.490\linewidth]{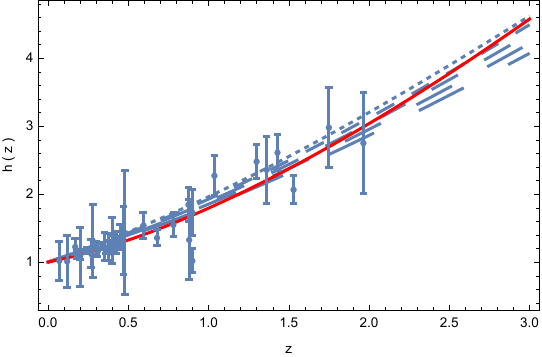} %
\includegraphics[width=0.490\linewidth]{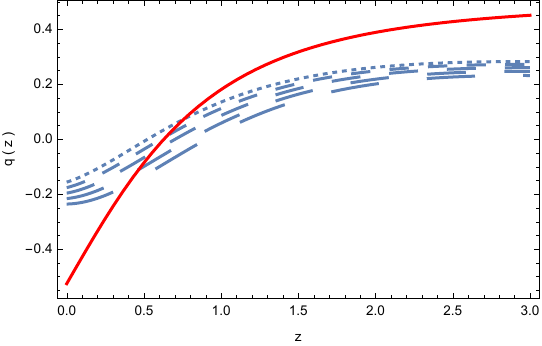}
\caption{Variations of the dimensionless Hubble function $h(z) $(left panel), and of the deceleration parameter $q(z)$ (right panel) for the hyperfluid model with initial conditions $u(0)=0.517$ and several values of $\Psi(0)$: $\Psi(0)=0.005$ (dotted curve), $\Psi(0)=0.015$ (short dashed curve), $\Psi(0)=0.025$ (dashed curve), $\Psi(0)=0.035$ (long dashed curve), $\Psi(0)=0.045$ (ultra-long dashed curve). The observational data for the Hubble function are represented with their error bars, while the red curve depicts the predictions of the $\Lambda$CDM model.}
\label{fig5}
\end{figure*}
\begin{figure*}[htbp]
\centering
\includegraphics[width=0.490\linewidth]{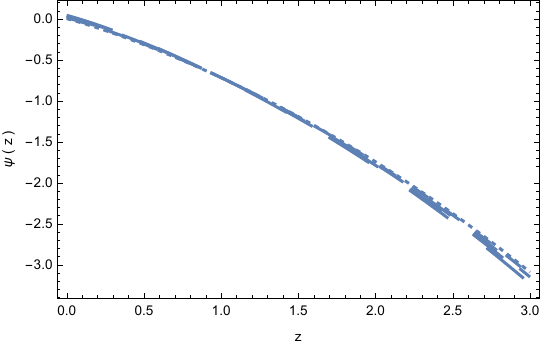} %
\includegraphics[width=0.490\linewidth]{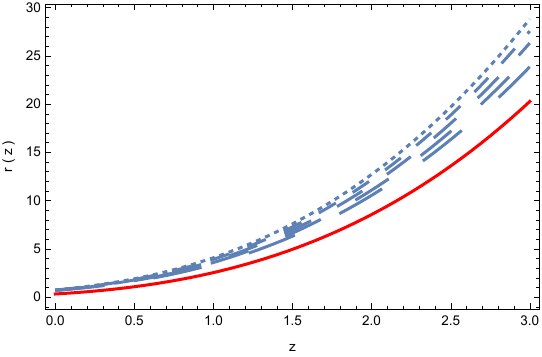}
\caption{Variations of the dimensionless non-metricity $\Psi(z)$ (left panel), and of the matter density (right panel) for the hyperfluid model with initial conditions $u(0)=0.517$ and several values of $\Psi(0)$: $\Psi(0)=0.005$ (dotted curve), $\Psi(0)=0.015$ (short dashed curve), $\Psi(0)=0.025$ (dashed curve), $\Psi(0)=0.035$ (long dashed curve), $\Psi(0)=0.045$ (ultra-long dashed curve). The red curve represents the predictions of the $\Lambda$CDM model.}
\label{fig6}
\end{figure*}

From Figure \ref{fig5} one can observe that for the considered range of parameters, the hyperfluid model describes the data of the Hubble function up to $z=2$. Depending on the initial conditions chosen, for higher redshifts the Hubble function of the model is either below or above the $\Lambda$CDM curve. On the same figure in the right panel, the deceleration parameter of the model can be seen, which behaves differently compared to the one of $\Lambda$CDM. Only in the range $0.5 <z < 0.7$ are the predictions of the two models close for the deceleration parameter.

The non-metricity scalar $\Psi(z)$, which can be seen on the left panel of Figure \ref{fig6}, is a monotone decreasing function of the redshift, and it takes negative values during the cosmological evolution. Moreover, its behaviour does not highly depend on the small shift in the initial conditions $\Psi(0)$. From the right panel, one can observe that the hyperfluid model predicts a slightly larger matter density than the standard $\Lambda$CDM model. The deviation is increasing with redshift and highly depends on the initial conditions chosen.

\begin{figure}[htbp]
\centering
\includegraphics[width=0.5\linewidth]{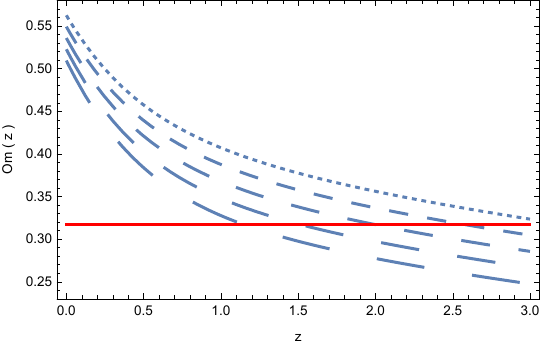}
\caption{Behavior of the function $Om (z)$ for the hyperfluid model with initial conditions $u(0)=0.517$ and several values of $\Psi(0)$: $\Psi(0)=0.005$ (dotted curve), $\Psi(0)=0.015$ (short dashed curve), $\Psi(0)=0.025$ (dashed curve), $\Psi(0)=0.035$ (long dashed curve), $\Psi(0)=0.045$ (ultra-long dashed curve). The predictions of the $\Lambda$CDM model are
represented by the red curve.}
\label{fig7}
\end{figure}

The $Om(z)$ diagnostic function, which is depicted in Figure \ref{fig7} shows a different behaviour compared both to the standard $\Lambda$CDM and Model $1$. In this case, $Om(z)$ is a monotone decreasing function of the redshift, which indicates a quintessence-like evolution. The quintessence-like behaviour seems to be independent of the initial condition chosen for $\Psi(0)$ in the range $0.005<\Psi<0.045$.

\section{Summary and Outlook}\label{section5}
In the present work, we described in a geometric, coordinate-free manner a special class of non-metric affine connections, which were first considered by Schrödinger. These connections, even though being metric-incompatible, preserve the lengths of vectors under autoparallel transport. In general, the non-metricity of a geometry can change lengths, angles, and volumes. It is important to note that there are other types of non-metric connections that preserve one of these three quantities, e.g., Weyl type preserves angles while equi-affine connections preserve volumes (see Proposition 2.1. in \cite{MATSUZOE2006567}).

After outlining the basic properties of Schr\"odinger, we gave examples of such connections, which were realized either through the torsion, non-metricity, or both, of a distinct affine connection. We obtained explicit formulae for the curvature tensors of the Yano-Schrödinger connection, which is closely related to the semi-symmetric connection, first introduced by Friedmann. With the help of the  semi-symmetric connection, we provided, for the first time in the literature, an explicit example of a non-static Einstein manifold with a vectorial torsion. For a general Schrödinger-type geometry, in which the non-metricity is not necessarily vectorial, we derived the Raychaudhuri and the Sachs equations. Then, we presented a Lagrangian formulation of the Sachs equation which was possible due to the length-preserving nature of the non-metricity.

Our results have laid a foundation for the development of a fully metric-affine gravity theory with a length-preserving non-metricity. In this paper, we worked in the Palatini formalism and assumed the connection to be symmetric from the beginning. In a completely metric-affine setting, one would have to impose the non-metricity and torsion conditions using Lagrange multipliers. By solving the constraint equations, one could obtain the desired type of non-metricity, such as the one we have worked with here.

To demonstrate the physical applicability of our analysis, we proposed two geometric generalizations of general relativity, in which the effects of dark energy have a firm geometric foundation. 
Both of these employed the Yano-Schrödinger connection, the only difference being the approach: first, we took a metric approach, and then, a metric-affine approach. In the second approach, we obtained the non-metricity by solving the connection field equations and presented characteristics of a novel Yano-Schrödinger hyperfluid. Similar analysis with torsion was done in \cite{Barrow_2019} and with non-metricity in \cite{Iosifidis_2023}. However, the latter did not consider a length-preserving non-metricity.

After this, we presented preliminary cosmological analysis using a spatially flat FLRW metric. We used the conservation of the energy-momentum to provide dynamics for the non-metricity. By comparing our cosmological models with observational data, we put tentative bounds on the initial values of the non-metricity vector. Moreover, we also showed that the Hubble parameter within our models matches the one in $\Lambda$CDM for the considered data set.

Nevertheless, we would like to note that our cosmological results have limitations. To confirm or disprove the validity of Yano-Schrödinger gravity and the Yano-Schrödinger hyperfluid, a detailed analysis using Monte-Carlo methods for a large number of data sets is necessary. Such data driven phenomenological studies have been carried out in other metric-affine extensions of GR \cite{Mandal_2023}. The physically relevant nature of Yano-Schr\"odinger gravity warrants similar analysis within this novel theory. 

Since it has been historically ignored, there exists a wide range of possible further geometric/physical analyses with Schrödinger connection. It remains an open question, whether any results from study of semi-symmetric connections could be translated or utilized in the theory of Yano-Schrödinger connections, which are realized by a semi-symmetric type of torsion. We hope that the present work will steer the interest of mathematicians and physicists alike towards the rich field of Schr\"odinger connections. 

\appendix
\section{Curvature Tensors}\label{appendixA}
This appendix is devoted to computing the curvature tensor of a general affine connection of the form
\begin{equation}
    \nabla_{X}Y=\overset{\circ}{\nabla}_{X} Y + U(-,X,Y).
\end{equation}
\begin{proposition}
    The Riemann curvature tensor of an affine connection $(\nabla,U)$ of the form
    \begin{equation}
    \nabla_{X}Y=\overset{\circ}{\nabla}_{X} Y + U(-,X,Y).
\end{equation}
is given by 
    \begin{multline*}
        Riem(\omega,Z,X,Y)=\overset{\circ}{Riem}(\omega,Z,X,Y)
+\left(\overset{\circ}{\nabla}_{X} U\right) (\omega,Y,Z)
-\left(\overset{\circ}{\nabla}_{Y}U \right)(\omega,X,Z) +U(\omega,X,U(-,Y,Z))\\
- U(\omega,Y,U(-,X,Z)).
    \end{multline*}
    where $\overset{\circ}{{Riem}}$ denotes the Riemann tensor of the Levi-Civita connection.
\end{proposition}
\begin{proof}
The Riemann tensor of an affine connection $\nabla$ is given by
\begin{equation*}
    Riem(\omega,Z,X,Y)=\omega\left( \nabla_X \nabla_Y Z-\nabla_Y \nabla_X Z-\nabla_{[X, Y]} Z \right)
\end{equation*}
To streamline the notation, we introduce the expression
\begin{equation*}
Riem(\omega,Z,X,Y)=\omega\left(\operatorname{Riem}(-,Z,X,Y) \right).
\end{equation*}
In the following, we evaluate $\operatorname{Riem}(-,X,Y,Z)$ through a series of steps. We begin with the definition, then expand it in two steps using the definition of a Schrödinger connection:
\begin{multline*}
  \operatorname{Riem}(-,Z,X,Y) =  \nabla_X \nabla_Y Z-\nabla_Y \nabla_X Z-\nabla_{[X, Y]} Z
  =  \nabla_X\left(\overset{\circ}{\nabla}_Y Z+U(-,Y, Z)\right) 
 -\nabla_Y\left(\overset{\circ}{\nabla}_X Z+U(-,X, Z)\right)\\  -\overset{\circ}{\nabla}_{[X, Y]} Z-U(-,[X, Y], Z) 
=  \overset{\circ}{\nabla}_X\left(\overset{\circ}{\nabla}_Y Z+U(-,Y, Z)\right)+U\left(-,X, \overset{\circ}{\nabla}_Y Z+U(-,Y, Z)\right)\\
 -\overset{\circ}{\nabla}_Y\left(\overset{\circ}{\nabla}_X Z+U(-,Y, Z)\right)
 -U\left(-,Y, \overset{\circ}{\nabla}_X Z+U(-,X, Z)\right) 
-\overset{\circ}{\nabla}_{[X, Y]} Z-U(-,[X, Y], Z).
\end{multline*}
On the right-hand side, we expand  the brackets using  multilinearity of $U$:
\begin{multline*}
         \textcolor{blue}{\overset{\circ}{\nabla}_X \overset{\circ}{\nabla}_Y Z}+\overset{\circ}{\nabla}_X(U(-,Y, Z))+U\left(-,X, \overset{\circ}{\nabla}_Y Z\right)+U(-,X, U(-,Y, Z)) 
\textcolor{blue}{-\overset{\circ}{\nabla}_Y \overset{\circ}{\nabla}_X Z}\\-\overset{\circ}{\nabla}_Y(U(-,X, Z))-U\left(-,Y, \overset{\circ}{\nabla}_X Z\right)-U(-,Y, U(-,X, Z)) 
 \textcolor{blue}{-\overset{\circ}{\nabla}_{[X, Y]} Z}-U(-,[X, Y], Z).
\end{multline*}
The combined terms in blue can be identified with the Riemann tensor of the Levi-Civita connection, hence
\begin{multline*}
         \overset{\circ}{\operatorname{Riem}}(-,Z,X,Y)  +\textcolor{red}{\overset{\circ}{\nabla}_X(U(-,Y, Z))}+U\left(-,X, \overset{\circ}{\nabla}_Y Z\right)+U(-,X, U(-,Y, Z)) \textcolor{red}{-\overset{\circ}{\nabla}_Y(U(-,X, Z))}\\-U\left(-,Y, \overset{\circ}{\nabla}_X Z\right)
-U(-,Y, U(-,X, Z)) -U(-,[X, Y], Z).
\end{multline*}
Applying the Leibniz Rule to the terms highlighted in red yields
\begin{multline*}
       \overset{\circ}{\operatorname{Riem}}(-,Z,X,Y)
+\left(\overset{\circ}{\nabla}_{X} U\right)(-,Y,Z)\textcolor{OliveGreen}{+U\left(-,\overset{\circ}{\nabla}_{X}Y,Z\right)}
\textcolor{Plum}{+U\left(-,Y,\overset{\circ}{\nabla}_{X}Z\right)}
\textcolor{orange}{+U \left(-,X,\overset{\circ}{\nabla}_{Y}Z \right)}\\+U(-,X,U(-,Y,Z)) 
 -\left( \overset{\circ}{\nabla}_{Y} U \right)(-,X,Z) \textcolor{OliveGreen}{-U \left(-, \overset{\circ}{\nabla}_{Y}X,Z \right)} \textcolor{orange}{- U \left(-,X,\overset{\circ}{\nabla}_{Y}Z \right)}
\textcolor{Plum}{-U\left(-,Y, \overset{\circ}{\nabla}_X Z\right)}\\-U(-,Y, U(-,X, Z)) \textcolor{OliveGreen}{-U(-,[X, Y], Z)}.
\end{multline*}
The green terms disappear due to the torsion-freeness of the Levi-Civita connection. The purple and orange terms cancel out, leaving us with
\begin{multline*}
        \operatorname{Riem}(-,Z,X,Y)=\overset{\circ}{\operatorname{Riem}}(-,X,Y,Z)
+\left(\overset{\circ}{\nabla}_{X} U\right)(-,Y,Z)
-\left(\overset{\circ}{\nabla}_{Y}U \right)(-,X,Z)+U(-,X,U(-,Y,Z))\\
- U(-,Y,U(-,X,Z)).
\end{multline*}
Therefore, for the Riemann tensor, we achieve the desired result
\begin{multline*}
    Riem(\omega,Z,X,Y)=\overset{\circ}{Riem}(\omega,Z,X,Y) + \left( \overset{\circ}{\nabla}_{X}U \right)(\omega,Y,Z)
    - \left(\overset{\circ}{\nabla}_{Y}U \right)(\omega,X,Z)+U(\omega,X,U(-,Y,Z))\\
    -U(\omega,Y,U(-,X,Z)).
\end{multline*}
\end{proof}

\begin{corollary}
    The Ricci curvature of an affine $(\nabla,U)$ is 
    \begin{multline}\label{schrodingerriccitensorcoordinatefree}
        Ric(Z,Y)=\overset{\circ}{Ric}(Z,Y)+ \left(\overset{\circ}{\nabla}_{X} U \right) \left( X^{\flat},Y,Z \right) 
        - \left(\overset{\circ}{\nabla}_{Y} U \right) \left( X^{\flat},X,Z \right) + U\left(X^{\flat},X,U(-,Y,Z) \right)\\
        -U\left(X^{\flat},Y,U(-,X,Z) \right).
    \end{multline}
\end{corollary}
\begin{corollary}
The Ricci scalar of an affine connection $(\nabla,U)$ is
     \begin{multline}\label{schrodingerricciscalarcoordinatefree}
        R=\overset{\circ}{R}+ \left(\overset{\circ}{\nabla}_{X} U \right) \left( X^{\flat},Y,Y \right) 
        - \left(\overset{\circ}{\nabla}_{Y} U \right) \left( X^{\flat},X,Y \right) + U\left(X^{\flat},X,U(-,Y,Y) \right)\\
        -U\left(X^{\flat},Y,U(-,X,Y) \right).
    \end{multline}
\end{corollary}
The local formulae for the Ricci tensor and Ricci scalar can be obtained by choosing a chart. Let $X=\partial_{\alpha},Y=\partial_{\nu},Z=\partial_{\mu}$ and consider equation \eqref{schrodingerriccitensorcoordinatefree}. This implies
\begin{multline*}
Ric(\partial_{\mu},\partial_{\nu})=\overset{\circ}{Ric}(\partial_{\mu},\partial_{\nu}) + \left(\overset{\circ}{\nabla}_{\partial_\alpha} U \right) \left(dx^{\alpha},\partial_{\nu},\partial_{\mu} \right)
    -\left(\overset{\circ}{\nabla}_{\partial_\nu} U\right) \left(dx^{\alpha},\partial_{\alpha},\partial_{\mu} \right)+U \left(dx^{\alpha},\partial_{\alpha},U(-,\partial_{\nu},\partial_{\mu}) \right) \\
    -U\left(dx^{\alpha},\partial_\nu ,U(-,\partial_\alpha,\partial_\mu) \right).
\end{multline*}
Using the definition of components in a chart, we obtain
\begin{equation}\label{schrodingerricci}
    R_{\mu \nu}=\overset{\circ}{R}_{\mu \nu} +\overset{\circ}{\nabla}_{\alpha} \tensor{U}{^\alpha _\nu _\mu}-\overset{\circ}{\nabla}_{\nu} \tensor{U}{^\alpha _\alpha _\mu}+\tensor{U}{^\alpha_\alpha  _\rho} \tensor{U}{^\rho_\nu_\mu} - \tensor{U}{^\alpha _\nu _\rho}\tensor{U}{^\rho_\alpha_\mu}.
\end{equation}

 \subsection{Coordinate Expressions for Curvature Tensors}\label{coordcurv}
Shortly put, the new results of this section are formulated in the following theorem.
\begin{theorem}\label{curvatureproperties}
    The Yano-Schrödinger connection $(\nabla,U)$ on a semi-Riemannian manifold has the following curvature properties:
    \begin{enumerate}
        \item[$(i)$] Ricci tensor:  $R_{\mu \nu}=\overset{\circ}{R}_{\mu \nu}+ g_{\mu \nu} \overset{\circ}{\nabla}_{\alpha} \pi^{\alpha} - \frac{1}{2} \overset{\circ}{\nabla}_{\mu} \pi_{\nu} + \overset{\circ}{\nabla}_{\nu} \pi_{\mu} - \frac{1}{2} g_{\mu \nu} \pi^{\alpha} \pi_{\alpha} - \frac{1}{4} \pi_{\nu} \pi_{\mu},$
         \item[$(ii)$] Ricci scalar: $ R=\overset{\circ}{R}+ \frac{9}{2} \overset{\circ}{\nabla}_{\mu} \pi^\mu - \frac{9}{4} \pi_{\alpha} \pi^{\alpha}$.
    \end{enumerate}
    
\end{theorem}
Before proving the theorem, we show that the coordinate-free definition of a Schrödinger connection  \eqref{schrfree} reproduces the well-known physics formula \eqref{Schrodingerconsideration}. Choosing a local chart $X=\partial_{\nu},Y=\partial_{\mu}$ and employing the definition of the connection coefficients yields
\begin{equation}
\begin{aligned}
    \tensor{\Gamma}{^\rho}_{\mu \nu} \partial_{\rho}&=\tensor{\gamma}{^\rho_\mu_\nu} \partial_{\rho}+ U(-,\partial_{\nu},\partial_\mu).
\end{aligned}
\end{equation}
To extract the components, we act with a covector field $dx^{\lambda}$, dual to the vector fields
\begin{equation}\label{con}
    \tensor{\Gamma}{^\lambda _\mu _\nu}=\tensor{\gamma}{^\lambda _\mu _\nu} + \tensor{U}{^\lambda_\nu_\mu}.
\end{equation}
We relate $U$ and $S$ in coordinates, by expressing the coordinate-free definition of the Schrödinger tensor in the same chart
\begin{equation}
    S \left(\partial_{\lambda}, \partial_{\nu},\partial_{\mu} \right)=U \left(g(\partial_{\lambda},-),\partial_{\nu}, \partial_{\mu} \right) \iff S_{\lambda \nu \mu}=g_{\lambda \beta} \tensor{U}{^\beta _\nu _\mu} \iff g^{\lambda \rho} S_{\lambda \nu \mu}=\tensor{U}{^\rho_ \nu _\mu}.
\end{equation}
Hence, equation \eqref{con} can be written as
\begin{equation}
    \tensor{\Gamma}{^\lambda _\mu _\nu}=\tensor{\gamma}{^\lambda _\mu _\nu} + g^{\rho \lambda} S_{\rho \nu \mu}.
\end{equation}
It's important to note that we cannot yet conclude that $S_{\rho \nu \mu}$ corresponds to the Schrödinger tensor in physics.  Proposition \ref{Schrodingersymmetry} in the local chart $A=\partial_{\rho},X=\partial_{\nu},Y=\partial_{\mu}$ directly implies
\begin{equation*}
    S(\partial_{\rho},\partial_{\nu},\partial_{\mu})=S(\partial_{\rho},\partial_{\mu},\partial_{\nu}), \; \;S(\partial_{\rho},\partial_{\nu},\partial_\mu)+S(\partial_\mu,\partial_\rho, \partial_\nu)+S(\partial_\nu,\partial_\mu,\partial_\rho)=0,
\end{equation*}
which is in perfect accordance with \eqref{Schrodingerconsideration}
\begin{equation*}
    S_{\rho \nu \mu}=S_{\rho \mu \nu}, \; \; S_{(\rho\nu \mu)}=0.
\end{equation*}
Hence, our approach reproduces the general framework of Schrödinger connections present in the physics literature. We now turn our attention to the special case of the Yano-Schrödinger connection. Recall, this connection is given by
\begin{equation}
     U(\omega,X,Y)=\frac{1}{2} \left(X^{\flat} \left( \pi \left( \omega^{\sharp} \right) Y-\pi(Y) \omega^{\sharp} \right) + Y^{\flat} \left(\pi \left( \omega^{\sharp}\right)X -\pi(X) \omega^{\sharp}\right)\right).
\end{equation}
Choosing a local chart $\omega=dx^{\alpha},X=\partial_{\mu},Y=\partial_{\nu}$ leads to
\begin{equation}\label{connectionmine}
\begin{aligned}
U\left(dx^{\alpha},\partial_{\mu},\partial_{\nu} \right)=g_{\mu \nu} \pi^{\alpha} - \frac{1}{2} \left(\delta^{\alpha}_{\mu} \pi_{\nu} + \delta^{\alpha}_{\nu} \pi_{\mu} \right).
\end{aligned}
\end{equation}
which is in accordance with the result obtained in \cite{Iosifidisthesis} in a completely different manner. In some more detail, after raising indices in formula $(1.87)$ of \cite{Iosifidisthesis}, and taking care of sign conventions, one obtains  \eqref{connectionmine}. As we will consider physical applications later on, we will interchangeably use the notations $ \tensor{U}{^\alpha_\mu_\nu}, \; \; \tensor{Q}{^\alpha_\mu_\nu}$ for a Yano-Schrödinger connection of the form \eqref{connectionmine}, and we might call this type of connection a length-preserving non-metricity.

We now move on to proving theorem \ref{curvatureproperties} by computing the Ricci tensor and Ricci scalar, which was not computed in the reference. In our convention, the Ricci tensor can be obtained by formally substituting \eqref{connectionmine} into \eqref{schrodingerricci}, and simplifying. Given the Ricci tensor, we obtain the Ricci scalar by contracting with the inverse metric. The computation for the Ricci tensor reads

\allowdisplaybreaks
\begin{align}\label{riccitensoryanoschrodinger}
    R_{\mu \nu}
    &=\overset{\circ}{R}_{\mu \nu} + g_{\mu \nu} \overset{\circ}{\nabla}_{\alpha} \pi^{\alpha} - \frac{1}{2} \overset{\circ}{\nabla}_{\mu} \pi_{\nu} + \overset{\circ}{\nabla}_{\nu} \pi_{\mu} - \frac{3}{2}g_{\mu \nu} \pi_{\rho} \pi^{\rho} + \frac{3}{4} \pi_{\mu} \pi_{\nu} + \frac{3}{4} \pi_{\nu} \pi_{\mu} - \pi_{\nu} \pi_{\mu} \nonumber\\
    &+\frac{1}{2} \left( \pi^{\alpha} \pi_{\alpha} g_{\mu \nu} + \pi_{\nu} \pi_{\mu} + \pi_{\nu} \pi_{\mu} + g_{\mu \nu}\pi_{\rho} \pi^{\rho} \right)-\frac{1}{4} \left( \pi_{\nu} \pi_{\mu}+4 \pi_{\nu} \pi_{\mu} +\pi_{\nu} \pi_{\mu} + \pi_{\nu} \pi_{\mu} \right) \nonumber\\
    &=\overset{\circ}{R}_{\mu \nu} + g_{\mu \nu} \overset{\circ}{\nabla}_{\alpha} \pi^{\alpha} - \frac{1}{2} \overset{\circ}{\nabla}_{\mu} \pi_{\nu} +\overset{\circ}{\nabla}_{\nu} \pi_{\mu} - \frac{1}{2}g_{\mu \nu} \pi^{\alpha} \pi_{\alpha} - \frac{1}{4}\pi_{\nu} \pi_{\mu}.
\end{align}

Hence, the Ricci scalar takes the form
\begin{equation}\label{ricciscalaryanoschrodinger}
    R=\overset{\circ}{R} + 4 \overset{\circ}{\nabla}_{\alpha} \pi^{\alpha} - \frac{1}{2} \overset{\circ}{\nabla}_{\mu} \pi^{\mu} + \overset{\circ}{\nabla}_{\nu} \pi^{\nu} - 2 \pi^{\alpha} \pi_{\alpha} - \frac{1}{4} \pi_{\nu} \pi^{\nu} =\overset{\circ}{R} + \frac{9}{2} \overset{\circ}{\nabla}_{\mu} \pi^{\mu} - \frac{9}{4} \pi_{\alpha} \pi^{\alpha}.
\end{equation}

\section{Derivation of Friedmann Equations}\label{appendixB}
In this appendix we derive the generalized Friedmann equations in Schrödinger-Yano gravity, using the assumptions from section \ref{CosmologyFriedmann}. As a starting point, recall that in the standard Riemannian geometry for the Levi-Civita connection the components of the Ricci tensor are given by
\begin{equation}
    \overset{\circ}{R}_{00}=-3\frac{\ddot a}{a}, \; \; \overset{\circ}{R}_{11}=\overset{\circ}{R}_{22}=\overset{\circ}{R}_{33}=a \ddot{a} +2 \dot{a}^2.
\end{equation}
The Ricci scalar takes the well-known form
\begin{equation}
    \overset{\circ}{R}=6\left(\frac{\ddot a}{a}+\frac{\dot a^2}{a^2} \right), 
\end{equation}
while the non-vanishing Christoffel symbols read
\begin{equation}
    \tensor{\gamma}{^0_i_j}=a \dot{a} \delta_{ij},
\quad
    \tensor{\gamma}{^i_0_j}=\frac{\dot a}{a} \delta^{i}_{j}, \; \; i,j=1,2,3.
\end{equation}
The generalized Einstein equation \eqref{einsteinequation} in the $00$ component takes the form
\begin{equation}\label{einstein00}
    \overset{\circ}{R}_{00}-\frac{1}{2} g_{00} \overset{\circ}{R} -\frac{5}{4} g_{00} \overset{\circ}{\nabla}_{\alpha}  \pi^{\alpha} + \frac{1}{4} \left(\overset{\circ}{\nabla}_{0} \pi_{0}+\overset{\circ}{\nabla}_{0} \pi_0 \right) +\frac{5}{8} g_{00} \pi^{\alpha} \pi_{\alpha} - \frac{1}{4} \pi_0 \pi_0=8 \pi T_{00}
\end{equation}
Before evaluating the above equation, we will compute some terms that appear explicitly. With our assumptions, we have
\begin{equation}
    \pi_{0}=\Pi, \; \; \pi^{0}=-\Pi, \; \; T_{00}=\rho,
\end{equation}
from which we get
\begin{equation}
    \overset{\circ}{\nabla}_0 \pi_0=\dot{\Pi}.
\end{equation}
The covariant divergence of $\pi$ is given by
\begin{equation}
    \overset{\circ}{\nabla}_{\alpha} \pi^{\alpha}=-\dot \Pi - 3 \frac{\dot a}{a} \Pi.
\end{equation}
Substituting everything back into \eqref{einstein00} we obtain
\begin{equation}
    -3\frac{\ddot a}{a}+\frac{1}{2}6\left(\frac{\ddot a}{a}+\frac{\dot a^2}{a^2} \right)+\frac{5}{4}\left(-\dot \Pi -3 \frac{\dot a}{a} \Pi \right) +\frac{1}{4}\left(2 \dot \Pi \right)+\frac{5}{8} \Pi^2 - \frac{1}{4} \Pi^2=8\pi \rho.
\end{equation}
A straightforward algebraic simplification leads to
\begin{equation}
    3\frac{\dot a^2}{a^2}-\frac{3}{4} \dot \Pi -\frac{15}{4} \frac{\dot a}{a} \Pi + \frac{3}{8} \Pi^2=8 \pi \rho.
\end{equation}
By introducing the Hubble function $H=\frac{\dot a}{a}$, we have the first Friedmann equation
\begin{equation}
    3H^2=8\pi \rho +\frac{3}{4} \dot \Pi +\frac{15}{4} H \Pi - \frac{3}{8} \Pi^2.
\end{equation}
For the second Friedmann equation, we consider the $ii$ components of the Einstein equation \eqref{einsteinequation}
\begin{equation}\label{einsteinii}
    \overset{\circ}{R}_{ii}-\frac{1}{2} g_{ii} \overset{\circ}{R} -\frac{5}{4} g_{ii} \overset{\circ}{\nabla}_{\alpha}  \pi^{\alpha} + \frac{1}{4} \left(\overset{\circ}{\nabla}_{i} \pi_{i}+\overset{\circ}{\nabla}_{i} \pi_i \right) +\frac{5}{8} g_{ii} \pi^{\alpha} \pi_{\alpha} -\frac{1}{4} \pi_{i} \pi_{i}=8 \pi T_{ii}.
\end{equation}
Thanks to the symmetry the equations will take the same form for all spatial indices $i$. For simplicity, we consider $i=1$. From our conventions, it can be seen that
\begin{equation}
    \pi_1=0, \; \; T_{11}=pa^2.
\end{equation}
So for the Einstein equation in the $11$ component  we have
\begin{equation}
    a \ddot a +2 \dot{a}^2 -a^2\frac{1}{2}6\left(\frac{\ddot a}{a}+\frac{\dot a^2}{a^2} \right)-\frac{5}{4} a^2 \left( - \dot \Pi -3 \frac{\dot a}{a} \Pi \right)+\frac{1}{4} \left( 2(-a \dot a \Pi)\right)-\frac{5}{8}a^2 \Pi^2=8\pi p a^2.
\end{equation}
Multiplying the terms out yields
\begin{equation}
    a \ddot a +2 \dot{a}^2 -3 a \ddot a -3 \dot a^2 +\frac{5}{4}  a^2 \dot \Pi +\frac{15}{4} a \dot a \Pi -\frac{1}{2} a \dot a \Pi-\frac{5}{8}a^2 \Pi^2=8 \pi p a^2.
\end{equation}
Collecting the terms and dividing by $a^2$ results in
\begin{equation}
    -2\frac{\ddot a}{a}-\frac{\dot a^2}{a^2}+ \frac{5}{4} \dot{\Pi} + \frac{13}{4} \frac{\dot a}{a} \Pi - \frac{5}{8} \Pi^2=8\pi p,
\end{equation}
which can be equivalently rewritten as
\begin{equation}
    -2 \dot{H} -2H^2 - H^2 +\frac{5}{4} \dot \Pi +\frac{13}{4} H \Pi -\frac{5}{8} \Pi^2 =8\pi p.
\end{equation}
Thus the final form of the second Friedmann equation is obtained
\begin{equation}
    3H^2+2 \dot{H}=-8\pi p +\frac{5}{4} \dot \Pi +\frac{13}{4} H \Pi -\frac{5}{8} \Pi^2.
\end{equation}
\section{Connection Field Equations}\label{appendixC}
In this section we provide the steps, which lead to equation \eqref{connectionfieldequations}. First of all, recall that a length-preserving non-metricity, or as we called it, a Yano-Schrödinger connection is characterized by
\begin{equation}
    \tensor{Q}{^\alpha _\mu _\nu}=g_{\mu \nu} \pi^{\alpha}-\frac{1}{2} \left( \delta^{\alpha}_{\mu} \pi_\nu + \delta^\alpha _\nu \pi_\mu \right).
\end{equation}
In a cosmological setting, it is characterized by a time-dependent function, and it can be written as
\begin{equation}
   \tensor{Q}{^\alpha _\mu _\nu}=g_{\mu \nu} D(t) u^{\alpha}-\frac{1}{2} \left( \delta^{\alpha}_{\mu} D(t) u_\nu + \delta^\alpha _\nu D(t) u_\mu \right).
\end{equation}
The corresponding non-metricity vectors read
\begin{equation}
    Q_{\lambda}=Q_{\lambda \mu \nu} g^{\mu \nu}=(n-1) D(t) u_{\lambda}, \; \; q_{\sigma}= Q_{\lambda \rho \sigma} g^{\lambda \rho}=\frac{-n+1}{2} D(t).
\end{equation}
Therefore, straightforward substitutions into the Palatini tensor yield
\begin{equation}
\begin{aligned}
     \tensor{P}{_\lambda^{(\rho \sigma)}}=&\frac{n-1}{2} D(t) u_{\lambda} g^{\rho \sigma} - D(t) u_\lambda g^{\rho \sigma} + \frac{1}{2} D(t) \delta^{\rho}_{\lambda} u^{\sigma} +\frac{1}{2} D(t) \delta^{\sigma}_{\lambda} u^{\rho}
     +\frac{1}{2}\frac{-n+1}{2}D(t)\left(u^\rho \delta^\sigma_\lambda + u^\sigma \delta^\rho_\lambda \right)\\
     &- \frac{1}{4} (n-1) D(t) \left( u^\rho \delta ^\sigma_\lambda + u^\sigma \delta^\rho_\lambda \right).
\end{aligned}
 \end{equation}
This expression can be simplified to obtain
\begin{equation}
     \tensor{P}{_\lambda^{(\rho \sigma)}}=\frac{n-3}{2}D(t) u_{\lambda} g^{\rho \sigma}+\frac{2-n}{2} D(t) \delta^{\rho}_{\lambda} u^{\sigma}+\frac{2-n}{2} D(t) \delta^{\sigma}_{\lambda} u^{\rho},
\end{equation}
or equivalently, by lowering indices we have
\begin{equation}
     P_{\lambda \mu \nu}
     =\frac{n-3}{2}D(t) u_{\lambda} g_{\mu \nu}+\frac{2-n}{2} D(t) g_{\lambda \mu} u_\nu+\frac{2-n}{2} D(t) g_{\nu \lambda} u_\mu.
 \end{equation}
 Employing the connection field equations $P_{\lambda \mu \nu}= \kappa \Delta_{\lambda \mu \nu}$ gives
 \begin{equation}
     \Delta_{\lambda \mu \nu}=\frac{D(t)}{2\kappa} \left( (n-3) u_\lambda g_{\mu \nu}+ (2-n) \left(g_{\lambda \mu} u_{\nu} + g_{\nu \lambda} u_{\mu} \right)\right).
 \end{equation}
 Using the $3+1$ decomposition of the metric in the cosmological setting, the hypermomentum can be equivalently written as
 \begin{equation}
      \Delta_{\lambda \mu \nu}=\frac{D(t)}{2\kappa} \left( (n-3) u_\lambda h_{\mu \nu}+ (2-n) \left(h_{\lambda \mu} u_{\nu} + h_{\nu \lambda} u_{\mu} \right)+(3-n+n-2+n-2) u_\lambda u_\mu u_\nu \right).
 \end{equation}
 Straightforward algebraic simplifications lead to the desired form
 \begin{equation}
         \Delta_{\lambda \mu \nu}=\frac{D(t)}{2 \kappa} \left(  (n-3) u_\lambda h_{\mu \nu} + (4-2n) h_{\lambda (\mu} u_{\nu)} +(n-1) u_\lambda u_\mu u_\nu \right) .
 \end{equation}
\section{Palatini Variation of the Hyperfluid Action} \label{appendixD}
To compute the variation of the action \eqref{actionhyp} with respect to the connection, we employ the identities
\begin{equation}\label{connectionidentities}
\begin{aligned}
    \delta_{\Gamma}{Q}_{\mu}=2 \delta^{(\rho}_{\mu} \delta^{\sigma)}_{\lambda} \delta \tensor{\Gamma}{^\lambda_\rho_\sigma}\;\; \text{and} \; \;
    \delta_{\Gamma} q_{\mu}=\left( g^{\rho \sigma} g_{\mu \lambda}+ \delta^{(\rho}_{\mu} \delta^{\sigma)}_{\lambda} \right) \delta \tensor{\Gamma}{^\lambda _\rho _\sigma}.
\end{aligned}
\end{equation}
Using \eqref{connectionidentities},  varying \eqref{actionhyp} with respect to the connection gives
\begin{equation}
\
    \frac{\delta I}{\delta \tensor{\Gamma}{^\lambda _\rho _\sigma}}=- \frac{\sqrt{-g}}{2} \Big[ (-3n+7) D^{(\rho} \delta^{\sigma)}_{\lambda}+(n-3)\left( g^{\rho \sigma} D_{\lambda}+ D^{(\rho} \delta^{\sigma)}_{\lambda}\right) \Big].
\end{equation}
Recollecting the terms, and using the definition of the hypermomentum, we obtain
\begin{equation}
    \tensor{\Delta}{_\lambda ^\rho ^\sigma}=(-2n+4) D^{(\rho} \delta^{\sigma)}_{\lambda}+(n-3) g^{\rho \sigma} D_{\lambda}=\frac{D}{2 \kappa} \Big[ (-2n+4) u^{(\rho} \delta^{\sigma)}_{\lambda} + (n-3) g^{\rho \sigma} u_{\lambda} \Big].
\end{equation}
Lowering the indices, we have
\begin{equation}
\begin{aligned}
    \tensor{\Delta}{_\lambda _\mu _\nu}&=g_{\rho \mu} g_{\sigma \nu} \tensor{\Delta}{_\lambda ^\rho ^\sigma}=\frac{D}{2 \kappa} \left[(n-3) g_{\mu \nu} u_{\lambda} + (4-2n) u_{(\mu} g_{\nu) \lambda} \right]\\
    &=\frac{D}{2 \kappa} \left[ (n-3) h_{\mu \nu} u_{\lambda}+(4-2n) h_{\lambda (\nu} u_{\mu)} +(3-n+2n-4) u_{\mu} u_{\nu} u_{\lambda} \right].
\end{aligned}
\end{equation}
Finally, a straightforward algebraic simplification yields the desired result
 \begin{equation}
     \Delta_{\lambda \mu \nu}=\frac{D}{2\kappa} \left[ (n-3) h_{\mu \nu} u_{\lambda} + (4-2n) h_{\lambda (\nu} u_{\mu)} +(n-1) u_\mu u_\nu u_\lambda \right].
 \end{equation}
\section*{Acknowledgments}
L.Cs. would like to thank Collegium Talentum Hungary and the StarUBB institute, for the research scholarships provided. Discussions with Prof. Tiberiu Harko and Tudor Patuleanu are highly appreciated. L.Cs. would also like to thank Aden Shaw for the guidance regarding creating illustrations. AA acknowledges support from an internal grant at the St. Mary's College of Maryland. DI's work was supported by the Estonian Research Council grant (SJD14).

\printbibliography

\end{document}